\documentclass[sigconf]{acmart}
\setcopyright{none}
\copyrightyear{}
\acmYear{}
\setcopyright{none}
\acmConference{}
\acmBooktitle{}
\acmPrice{}
\acmDOI{}
\acmISBN{}

\usepackage{xcolor}
\usepackage{tcolorbox}

\definecolor{mylinkcolor}{RGB}{0,0,140}

\PassOptionsToPackage{hyphens}{url}\usepackage{hyperref} 

\definecolor{mylinkcolor}{RGB}{0,0,140}
\hypersetup{colorlinks,allcolors=mylinkcolor,citecolor=mylinkcolor}
\usepackage[capitalize]{cleveref}

\usepackage{flushend}
\usepackage{amsfonts}
\usepackage{amsmath}
\usepackage{mathrsfs}  
\usepackage{graphicx}
\usepackage{booktabs}
\usepackage{enumitem}
\usepackage{algorithm}
\usepackage{algorithmicx}
\usepackage[noend]{algpseudocode}
\algdef{SE}[DOWHILE]{Do}{doWhile}{\algorithmicdo}[1]{\algorithmicwhile\ #1}%

\usepackage{enumitem}

\usepackage{mathtools}
\DeclarePairedDelimiter\ceil{\lceil}{\rceil}
\DeclarePairedDelimiter\floor{\lfloor}{\rfloor}

\newcommand{\cut}{\textbf{cut}}

\providecommand{\mat}[1]{\boldsymbol{\mathrm{#1}}}%
\renewcommand{\vec}[1]{\boldsymbol{\mathrm{#1}}}

\DeclareMathOperator*{\minimize}{minimize}

\DeclareMathOperator{\argmin}{argmin}

\providecommand{\mA}{\ensuremath{\mat{A}}}

\providecommand{\mD}{\ensuremath{\mat{D}}}

\providecommand{\mL}{\ensuremath{\mat{L}}}
\providecommand{\mM}{\ensuremath{\mat{M}}}

\providecommand{\ve}{\ensuremath{\vec{e}}}

\providecommand{\vw}{\ensuremath{\vec{w}}}
\providecommand{\vx}{\ensuremath{\vec{x}}}

\usepackage{xparse}

\newtheorem{problem}{Problem}

\usepackage{subfigure}
\begin{document}
	
\title{Cut-matching Games for Generalized Hypergraph Ratio Cuts}
	
	
	%
	
\begin{abstract}
	Hypergraph clustering is a basic algorithmic primitive for analyzing complex datasets and systems characterized by multiway interactions, such as group email conversations, groups of co-purchased retail products, and co-authorship data. This paper presents a practical $O(\log n)$-approximation algorithm for a broad class of hypergraph \emph{ratio cut} clustering objectives. This includes objectives involving generalized hypergraph cut functions, which allow a user to penalize cut hyperedges differently depending on the number of nodes in each cluster. Our method is a generalization of the cut-matching framework for graph ratio cuts, and relies only on solving maximum s-t flow problems in a special reduced graph. It is significantly faster than existing hypergraph ratio cut algorithms, while also solving a more general problem. In numerical experiments on various types of hypergraphs, we show that it quickly finds ratio cut solutions within a small factor of optimality.
\end{abstract}

	\author{Nate Veldt}
	\affiliation{
		\institution{Texas A\&M University}
		\department{Department of Computer Science and Engineering}
	}
	\email{nveldt@tamu.edu}

	\maketitle
	
	\section{Introduction}
Graphs are a popular way to model social networks and other datasets characterized by pairwise interactions, but there has been a growing realization that many complex systems of interactions are characterized by \emph{multiway} relationships that cannot be directly encoded using graph edges. For example, users on social media join interest groups and participate in group events, online shoppers purchase multiple products at once, discussions in Q\&A forums typically involve multiple participants, and email communication often involves multiple receivers. These higher-order interactions can be modeled by a \emph{hypergraph}: a set of nodes that share multiway relationships called hyperedges. A standard primitive for analyzing group interaction data is to perform clustering, i.e., identifying groups of related nodes in a hypergraph. This is a direct generalization of the well-studied \emph{graph} clustering problem. Hypergraph clustering methods have been used for many different data analysis tasks, such as detecting different types of Amazon products based on groups of co-reviewed products~\cite{veldt2020hyperlocal}, identifying related vacation rentals on Trivago from user browsing behavior~\cite{chodrow2021generative}, clustering restaurants based on Yelp reviews~\cite{liu2021strongly}, and finding related posts in online Q\&A forums from group discussion data~\cite{veldt2021approximate}.

A ``good'' cluster in a (hyper)graph is a set of nodes with many internal (hyper)edges, but few (hyper)edges connecting different clusters. One of the most standard ways to formalize this goal is to minimize a {ratio cut} objective, which measures the ratio between the cluster's \emph{cut} and some notion of the cluster's size. In the graph setting, the {cut} simply counts the number of edges across the cluster boundary. The standard hypergraph cut function similarly counts the number of hyperedges on the boundary.  However, unlike the graph setting, there is more than one way to partition a hyperedge across two clusters. This has led to the notion of a generalized hypergraph cut function, which assigns different penalties to different ways of splitting a hyperedge. The performance of downstream hypergraph clustering applications can strongly depend on the type of generalized cut penalty that is used~\cite{liu2021strongly,fountoulakis2021local,veldt2020hyperlocal,veldt2021approximate,veldt2022hypercuts,panli2017inhomogeneous,panli_submodular}. 

\paragraph{Previous work and existing limitations.} Common ratio cut objectives include conductance, normalized cut, and expansion, which have been defined both for graphs and hypergraphs (see Section~\ref{sec:prelims}).  These problems are NP-hard even in the graph case, but approximate solutions can be obtained using spectral methods, which often come with so-called \emph{Cheeger} inequality guarantees~\cite{chung1997spectral}, or by using various types of multicommodity-flow and expander embedding methods~\cite{arora2009expander,khandekar2009graph,leighton1999multicommodity,orecchia2008partitioning}. Several existing techniques for approximating graph ratio cuts have already been extended to hypergraphs, but there are many new challenges in this setting. First of all, many generalized notions of graph Laplacian operators are nonlinear, which leads to additional challenges in computing or approximating eigenvectors and eigenvalues to use for clustering~\cite{chan2018spectral,panli_submodular,yoshida2019cheeger}. Secondly, there are many challenges associated with obtaining guarantees for generalized hypergraph cut functions, which can be much more general than graph cut functions. 

Given the above challenges, there is currently a wide gap between theoretical algorithms and practical techniques for minimizing hypergraph ratio cut objectives. Many theoretical approximation algorithms rely on expensive convex relaxations and complicated rounding techniques that are not implemented in practice~\cite{chan2018spectral,panli_submodular,yoshida2019cheeger,louis2016approximation,kapralov2020towards}. This includes memory-intensive LP and SDP relaxations with $O(n^3)$ constraints~\cite{louis2016approximation,kapralov2020towards} (where $n$ is the number of nodes), other SDP relaxations with very large constraint sets involving Lov\'{a}sz extensions of submodular functions~\cite{panli_submodular,yoshida2019cheeger}, and rounding techniques whose implementation details rely on black-box subroutines with unspecified constants~\cite{louis2016approximation}. 
While these methods are interesting from a theoretical perspective, they are far from practical. 
Another downside specifically for spectral techniques~\cite{panli2017inhomogeneous,BensonGleichLeskovec2016,panli_submodular,chan2018spectral,yoshida2019cheeger,ikeda2022finding} is that their Cheeger-like approximation guarantees are $O(n)$ in the worst case, which is true even in the graph setting. Spectral techniques also often come with approximation guarantees that get increasingly worse as the maximum hyperedge size grows~\cite{BensonGleichLeskovec2016,panli2017inhomogeneous,chan2018spectral}. A limitation in another direction is that many existing hypergraph ratio cut algorithms are designed only for the simplest notion of a hypergraph cut function rather than generalized cuts~\cite{chan2018spectral,BensonGleichLeskovec2016,Zhou2006learning,kapralov2020towards,louis2016approximation,ikeda2022finding}. Finally, although a number of recent techniques apply to generalized hypergraph cut functions and come with practical implementations, these do not provide theoretical approximation guarantees for global ratio cut objectives because they are either heuristics~\cite{panli_submodular} or because they focus on minimizing localized ratio cut objectives that are biased to a specific region of a large hypergraph~\cite{veldt2020hyperlocal,liu2021strongly,ibrahim2020local,fountoulakis2021local,veldt2021approximate}.

\paragraph{The present work: practical and general approximation algorithms for hypergraph ratio cuts.} 
\begin{table*}[t]
	\caption{A summary of algorithms and guarantees for minimizing  hypergraph ratio cuts, where $n$ is the number of nodes. 
		There are differences among different \emph{Cheeger}-style approximation guarantees (e.g., how the approximation scales with the largest hyperedge size), but none are better than $O(n)$ in the worst case. 
		Submodular cardinality-based cut functions are significantly more general than the most basic \emph{all-or-nothing} function, and are arguably the most common submodular hypergraph cut functions used in practice. The previous methods listed here were developed for a specific node weight function corresponding to a specific ratio cut problem (e.g., expansion or conductance), whereas our approach is designed for arbitrary node weights.
		The last column indicates whether an implementation was provided when the method was first presented. Methods without an implementation come with various practical challenges and were designed primarily as theoretical results.
	}
	\label{tab:comp}
	\centering
	\begin{tabular}{l l l l l}
		\toprule 
		Method & Approx Guarantee & Hypergraph cut function & Main Technique & Implemented    \\
		\midrule
		CLTZ~\cite{chan2018spectral} & Cheeger & All-or-nothing & SDP relaxation & No \\
		Li-Milenkovic CE~\cite{panli2017inhomogeneous} & Cheeger  & Submodular & Clique expansion + graph spectral & Yes \\
		Li-Milenkovic IPM~\cite{panli_submodular} & N/A & Submodular & Inverse power method & Yes \\
		Li-Milenkovic SDP~\cite{panli_submodular} & Cheeger & Submodular & SDP relaxation & No \\
		IMTY HeatEQ~\cite{ikeda2022finding} & Cheeger & All-or-nothing & Solving hypergraph heat equation & No \\
		Louis-Makarychev~\cite{louis2016approximation} & $O(\sqrt{\log n})$ & All-or-nothing & SDP relaxation & No \\
		KKTY~\cite{kapralov2020towards} & $O(\log n)$ & All-or-nothing & LP relaxation & No \\
		{\bf This paper} & \textbf{$O(\log n)$} & \textbf{Submodular, cardinality-based} & Cut-matching (repeated max-flow) & \textbf{Yes }\\
		\bottomrule 
	\end{tabular}
	\end{table*}

This paper presents the first algorithm for hypergraph ratio cuts 
that comes with a nontrivial approximation guarantee (i.e., better than $O(n)$ in the worst case) and also applies to generalized hypergraph cuts. Additionally, compared with approximation algorithms applying only to the standard hypergraph cut function, our method comes with a substantially faster runtime guarantee and a practical implementation.
In more detail, our algorithm has a worst-case $O(\log n)$ approximation guarantee that is independent of the maximum hyperedge size. It applies to any cardinality-based submodular hypergraph cut function, which is a general and flexible hypergraph cut function that captures many popular hypergraph cut functions as special cases~\cite{veldt2022hypercuts}. It relies only on solving a sequence of maximum $s$-$t$ flow problems in a special type of reduced graph, which can be made efficient using an off-the-shelf max-flow subroutine. Table~\ref{tab:comp} provides a comparison of our method against the best previous algorithms for hypergraph ratio cuts in terms of approximation guarantees, the type of hypergraph cut functions they apply to, and whether or not a practical implementation exists.

Our method generalizes the \emph{cut-matching} framework for graph ratio cuts~\cite{khandekar2009graph,orecchia2008partitioning}, which alternates between solving max-flows and finding bisections in order to embed an expander into the graph. Although maximum flow problems and expander embeddings are well-understood in the graph setting, these concepts are more nuanced and not as well understood in hypergraphs, especially in the case of generalized hypergraph cuts. We overcome these technical challenges by providing a new approach for embedding an expander \emph{graph} inside a hypergraph, which is accomplished by solving maximum flow problems in a reduced directed graph that models the generalized cut properties of the hypergraph. To summarize: 
\begin{itemize}[itemsep =0pt,leftmargin = 10pt]
	\item We present a new framework for embedding expander graphs into hypergraphs in order to obtain lower bounds for NP-hard hypergraph ratio cut objectives (Lemma~\ref{lem:hypbound}).
	\item We present an $O(\log n)$-approximation algorithm that applies to hypergraph ratio cut objectives with any submodular cardinality-based cut function and any positive node weight function, and relies simply on solving a sequence of graph max-flow problems.
	\item We provide additional practical techniques for making our method more efficient in practice and for obtaining improved lower bounds and a posteriori approximation guarantees (Theorem~\ref{thm:bounds}).
	\item We implement our method and show that it is orders of magnitude more scalable than previous approximation algorithms and allows us to detect high quality ratio cuts in large hypergraphs from numerous application domains.
\end{itemize}

	\section{Preliminaries and Related Work}
\label{sec:prelims}
Let $G = (V,E)$ denote a graph with node set $V$ and edge set $E$. We denote edges in undirected graphs using bracket notation $e = \{u,v\} \in E$; an edge may additionally be associated with a nonnegative edge weight $w_G(e) = w_G(u,v) \geq 0$. Given a subset of nodes $S \subseteq V$, let $\bar{S} = V\backslash{S}$ denote the complement set. The boundary of $S$ is given by
\begin{equation}
	\label{eq:boundarys}
	\partial_G(S) = \partial_G (\bar{S}) = \{ e \in E \colon |e \cap S| = 1 \text{ and } |e \cap \bar{S}| = 1\}.
\end{equation}
This is the set of edges with one endpoint in $S$ and the other endpoint in $\bar{S}$. The \emph{cut} of $S$ in $G$ is the weight of edges crossing the boundary:
\begin{equation}
	\label{eq:cut}
	\cut_G(S) = \cut_G(\bar{S}) = \sum_{e \in \partial_G(S)} w_G(e).
\end{equation}
Given a positive node weight function $\pi \colon V \rightarrow \mathbb{R}_{> 0}$, the ratio cut objective with respect to $\pi$ for a set $S$ is denoted by
\begin{equation}
	\label{eq:expansion}
	\phi_G(S,\pi) = \frac{\cut_G(S)}{\min \{\pi(S) , \pi(\bar{S}) \}},
\end{equation}
where $\pi(S) = \sum_{u \in S} \pi(u)$. This captures the well-known special case of \emph{conductance} (when $\pi(u)$ equals node degree $d_u$), and \emph{expansion} (when $\pi(u)  = 1$). We will refer to $\phi_G(S,\pi)$ as the $\pi$-expansion of $S$. The $\pi$-expansion of a graph $G$ is then 
\begin{equation}
	\label{eq:minepx}
	\phi_{G,\pi} = \min_{S \subset V} \phi_G(S,\pi).
\end{equation}
We will drop $G$ and $\pi$ from subscripts when they are clear from context. A graph is known as an \emph{expander} graph if its expansion is at least a constant. Formally, for a given node weight function $\pi$, we say that $G = (V,E)$ is a $\pi$-expander if there exists a constant $c > 0$ such that for every $S \subseteq V$, $\phi(S) > c$. 



\textbf{Directed graphs.} Our results rely on notions of cuts in directed graphs. If $G = (V,E)$ is directed, we use parenthesis notation to denote edge directions, i.e. $(u,v) \in E$ indicates an edge from node $u$ to node $v$. For a set $S \subseteq V$, the directed cut function is 
\begin{equation*}
	{\cut}_G(S) = \sum_{(u,v)\in E \colon u \in S, v \in \bar{S}} w_G(u,v).
\end{equation*}
Expansion in a directed graph is defined by using a directed cut function in~\eqref{eq:expansion}. We can treat an undirected graph $G$ as a directed graph by replacing each undirected edge  with two directed edges, in which case the two notions of $\pi$-expansion match.
\subsection{Graph flows}
A flow function $f \colon E \rightarrow \mathbb{R}_{\geq 0}$ on a directed graph $G = (V,E)$ assigns a nonnegative flow value $f_{ij} \geq 0$ to each edge $(i,j) \in E$. If $f_{ij} \leq w_G(i,j)$ for each $(i,j) \in E$, then $f$ satisfies \emph{capacity constraints}. In general, the \emph{congestion} of $f$ is the maximum capacity violation:
\begin{equation}
	\label{eq:congestion-def}
	\textbf{congestion}(f) = \max_{(i,j) \in E} \; \frac{f_{ij}}{w_G(i,j)}.
\end{equation}
We say that $f$ satisfies flow constraints at a node $v \in V$ if the flow into $v$ equals the flow out of $v$:
\begin{equation}
	\sum_{i \colon (i,v) \in E} f_{iv} = \sum_{j \colon (v,j) \in E} f_{vj}.
\end{equation}
If the flow into $v$ is greater than the flow out of $v$ then $v$ is an \emph{excess} node. If the flow out of $v$ is more than the flow into $v$ then $v$ is a \emph{deficit} node.
Given two flow functions $f^{(1)}$ and $f^{(2)}$, the sum $f' = f^{(1)}  + f^{(2)}$ is the flow obtaining by defining $f'_{ij} = f^{(1)}_{ij} + f^{(2)}_{ij}$ for each $(i,j) \in E$. This flow satisfies $\textbf{congestion}(f') \leq \textbf{congestion}(f^{(1)}) + \textbf{congestion}(f^{(2)})$. If $f^{(1)}$ and $f^{(2)}$ both satisfy flow constraints at a node $v$, then $f'$ will as well. We will consider two special types of flows.
\begin{definition} [$s$-$t$ flow]
	Given $\{s,t\} \subseteq V$, an $s$-$t$ flow on $G$ is a flow that satisfies
	flow constraints on each $v \in V- \{s,t\}$.
	We say that $f$ routes flow from $s$ to $t$, and has a flow value
	\begin{equation}
		|f| = \sum_{i \colon (s,i) \in E} f_{si} - \sum_{j \colon (j,s) \in E} f_{js}.
	\end{equation}
\end{definition}

\begin{definition} [multicommodity flow]
	A multicommodity flow problem in $G$ is defined by a set $\mathcal{D}$ of demand pairs $(u,v) \in V \times V$ and corresponding weights $w_{uv} \geq 0$. The flow $f$ is a feasible multicommodity flow for $\mathcal{D}$ if it can be written as $f = \sum_{(u,v) \in \mathcal{D}} f^{(uv)}$ where $f^{(uv)}$ is a $u$-$v$ flow that routes $w_{uv}$ units of flow from $u$ to $v$.
	%
\end{definition}

We will use the following standard flow decomposition lemma. 
\begin{lemma}{(Theorem 3.5 in~\cite{ahuja1988network} or Lemma 2.20 in~\cite{williamson2019network}.)}
	\label{lem:flowdeomp}
	Let $f$ be a flow function on a graph $G = (V,E)$ with $n$ nodes and $m$ edges. The flow can be decomposed as $f = \sum_{i = 1}^\ell f_i$
	where $\ell \leq m + n$, and where for each $i$, the edges with positive flow in $f_i$ either form a simple directed path from a deficit node to an excess node or form a cycle. 
\end{lemma}


\subsection{General hypergraphs cuts and expansion }
A hypergraph $\mathcal{H} = (V,\mathcal{E})$ is a generalization of a graph where $V$ denotes a node set and $\mathcal{E}$ is a \emph{hyperedge} set where each $e \in \mathcal{E}$ is a subset of nodes in $V$ of arbitrary size. The \emph{order} of a hyperedge $e \in \mathcal{E}$ is the number of nodes it contains, $|e|$. The degree of a node $v$ is denoted by $d_v$. The boundary of $S\subseteq V$ is the set of hyperedges crossing a bipartition $\{S, \bar{S}\}$:
\begin{equation}
	\label{eq:hypbound}
	\partial_\mathcal{H}(S) = \partial_\mathcal{H} (\bar{S}) = \{ e \in \mathcal{E} \colon |e \cap S| > 0 \text{ and } |e \cap \bar{S}| > 0\}.
\end{equation}
If each hyperedge $e \in \mathcal{E}$ has a nonnegative scalar weight $w_\mathcal{H}(e)$, the standard \emph{all-or-nothing} hypergraph cut function is given by
\begin{equation}
	\label{eq:aon}
	\textstyle{\cut_\mathcal{H}(S) = \sum_{e \in \partial_\mathcal{H}(S)} w_\mathcal{H}(e)}.
\end{equation}
It is useful in many settings to assign different cut penalties for separating the nodes of a hyperedge in different ways. This has led to the concept of generalized hypergraph cut functions~\cite{yoshida2019cheeger,panli2017inhomogeneous,panli_submodular,veldt2022hypercuts,veldt2020hyperlocal,veldt2021approximate,fountoulakis2021local,liu2021strongly,zhu2022hypergraph}. Formally, each hyperedge $e \in \mathcal{E}$ is associated with a \emph{splitting function} $\vw_e \colon A \subseteq e \rightarrow \mathbb{R}$ that assigns a penalty for each way of separating the nodes of a hyperedge. The generalized hypergraph cut is then given by
\begin{equation}
	\label{eq:gencut}
	\cut_\mathcal{H}(S) = \sum_{e \in \partial_\mathcal{H}(S)} \vw_e(S \cap e).
\end{equation}
Following recent work~\cite{veldt2022hypercuts,veldt2020hyperlocal,veldt2021approximate,panli_submodular}, we focus on splitting functions satisfying the following properties for all $A,B \subseteq e$:
\begin{align*}
	(\text{nonnegative) } & \vw_e(A) \geq 0 \\
	(\text{uncut-ignoring) } & \vw_e(\emptyset) = 0   \\
	(\text{symmetric) } 	& \vw_e(A) = \vw(e \backslash A) \\
	(\text{submodular) }  & \vw_e(A) + \vw_e(B) \geq \vw_e(A \cap B) + \vw_e(A \cup B) \\
	(\text{cardinality-based) }  & \vw_e(A) = \vw_e(B) \text{ if $|A| = |B|$}.
\end{align*}
%
When these conditions are satisfied, $\cut_\mathcal{H}$ is a \emph{cardinality-based submodular hypergraph cut function}. To simplify terminology we will refer to it simply as a \emph{generalized hypergraph cut function}.

\paragraph{Generalized hypergraph cut expansion}
We are now ready to formalize the generalized ratio cut problem we are interested in solving over a hypergraph $\mathcal{H} = (V,\mathcal{E})$.
	Given a positive node weight function $\pi \colon V \rightarrow \mathbb{R}_{> 0}$ and a generalized hypergraph cut function $\cut_\mathcal{H}$, the hypergraph $\pi$-expansion of a set $S \subseteq V$ is defined to be.
	\begin{equation}
		\label{eq:hyperexpansion}
		\phi_\mathcal{H} (S,\pi) = \frac{\cut_\mathcal{H}(S)}{\min \{\pi(S) , \pi(\bar{S}) \}}.
	\end{equation}
	\begin{tcolorbox}
	\begin{problem}
	The hypergraph $\pi$-expansion problem is to solve
	\begin{equation*}
		\phi_{\mathcal{H}, \pi} = \min_{S \subseteq V} \frac{\cut_\mathcal{H}(S)}{\min \{\pi(S) , \pi(\bar{S}) \}}
	\end{equation*}
where $\cut_\mathcal{H}$ is a generalized hypergraph cut function and $\pi$ is a positive node weight function. 
	\end{problem}
\end{tcolorbox}
	Our goal is to develop an approximation algorithm for this objective that applies to submodular cardinality-based hypergraph cut functions and arbitrary node weight functions. For runtime analyses, we assume that when hyperedge cut penalties and node weights are scaled to be at least one, then all other node weights and cut penalties are polynomially bounded in terms of $|V|$ and $|\mathcal{E}|$.
	


\subsection{Related work}
Expansion ($\pi(v) = 1$) and conductance ($\pi(v) = d_v$) are two of the most widely-studied graph ratio cut objectives. 
For many years, the best approximation for graph expansion and conductance was $O(\log n)$, based on solving a multicommodity flow problem~\cite{leighton1999multicommodity} that can also be viewed as a linear programming relaxation. The current best approximation factor is $O(\sqrt{\log n})$, based on a semidefinite programming relaxation~\cite{arora2009expander}. Khandekar et al.~\cite{khandekar2009graph} later introduced the \emph{cut-matching} framework, providing an $O(\log^2 n)$ approximation based on faster maximum $s$-$t$ flow computations, rather than multicommodity flows. The same framework was later used to design improved $O(\log n)$-approximation algorithms~\cite{orecchia2008partitioning,orecchia2011fast}.

Many variants of the hypergraph ratio-cut objective~\eqref{eq:hyperexpansion} have previously been considered~\cite{louis2015hypergraph,veldt2022hypercuts,Zhou2006learning,panli2017inhomogeneous,yoshida2019cheeger,panli_submodular,liu2021strongly}, which consider different types of cut functions and node weight functions $\pi$. A number of Cheeger-style approximation guarantees have been developed for these types of objectives based on eigendecompositions of nonlinear hypergraph Laplacians~\cite{panli_submodular,yoshida2019cheeger,chan2018spectral}, but in the worst case these yield $O(n)$ approximation factors. Other techniques rely on clique expansions of hypergraphs~\cite{BensonGleichLeskovec2016,Zhou2006learning,panli2017inhomogeneous}, which also have approximation factors that are $O(n)$ in the worst case
and additionally scale poorly with the maximum hyperedge size. The $O(\log n)$ multicommodity flow algorithm~\cite{leighton1999multicommodity} and the SDP-based $O(\sqrt{\log n})$ approximation for expansion~\cite{arora2009expander} have both been generalized to the hypergraph setting~\cite{louis2016approximation,kapralov2020towards}. However, these only apply to the standard {all-or-nothing} hypergraph cut function~\eqref{eq:aon}, and focus exclusively on either the expansion node weight choice ($\pi(v) = 1$)~\cite{louis2016approximation} or the conductance node weight choice ($\pi(v) = d_v$)~\cite{kapralov2020towards}. Finally, there are numerous results on minimizing locally-biased ratio cut objectives in graphs~\cite{Andersen:2008:AIG:1347082.1347154,AndersenChungLang2006,Veldt2019flow,veldt16simple,orecchia2014flow,kloster2014heat,mahoney2012local} and hypergraphs~\cite{veldt2020hyperlocal,veldt2021approximate,liu2021strongly,fountoulakis2021local,ibrahim2020local}, but these do not provide guarantees for global ratio cut objectives.

\textbf{Concurrent work on hypergraph cut-matching.} In independent and concurrent work, Ameranis et al.~\cite{ameranis2023efficient} developed an $O(\log n)$-approximation algorithm for hypergraph ratio cuts with generalized cut functions, also by extending the cut-matching framework of Khandekar et al.~\cite{khandekar2009graph}. 
At a high-level, their approach for obtaining an $O(\log n)$-approximation is analogous to ours in that it relies on reducing the hypergraph to a graph and then solving maximum flow problems to embed a multicommodity flow problem in the hypergraph. On a more technical level, the strategies differ in terms of the graph reduction technique that is used and the maximum flow problem that needs to be solved. Ameranis et al.~\cite{ameranis2023efficient} use a bipartite graph expansion of the hypergraph, which replaces each hyperedge with a new vertex and adds undirected edges between original vertices and hyperedge vertices. They then solve maximum flow problems in \emph{undirected} flow graphs with vertex capacity constraints, in order to satisfy a notion of hypergraph flow they define. We instead use an augmented cut preserver (see Definition~\ref{def:augpres}) when reducing the hypergraph to a graph, and then repeatedly solve directed maximum flow problems \emph{without} vertex capacity constraints. 

In terms of theory, the algorithm designed by Ameranis et al.~\cite{ameranis2023efficient} applies to a new notion of a monotone submodular cut function, which is more restrictive than a general submodular cut function but generalizes the submodular cardinality-based setting we consider. Although their approach shows promise for practical implementations, the authors focused on theoretical contributions, and therefore did not implement the method or show numerical results. In contrast, one of the main contributions of our work is to provide a practical implementation for hypergraph cut-matching games and show detailed comparisons against other algorithms for hypergraph ratio cuts. Combining these strategies for hypergraph cut-matching in order to develop even better theoretical results and implementations is a promising direction for future research.

	\section{Hypergraph Expander Embeddings}
Expander embeddings are common techniques for lower bounding expansion in undirected graphs~\cite{leighton1999multicommodity,sherman2009breaking,arora2009expander,khandekar2009graph}. We generalize this basic approach in order to develop a strategy for lower bounding hypergraph ratio cut objectives, by embedding an expander {graph} into a special type of directed graph that models a generalized hypergraph cut function. 

\subsection{Hypergraph cut preservers}
Many hypergraph clustering methods rely on reducing a hypergraph to a graph and then applying an existing graph technique~\cite{panli2017inhomogeneous,veldt2022hypercuts,ihler1993modeling,BensonGleichLeskovec2016,yin2017local}. We use a precise previous notion of hypergraph reduction for modeling generalized hypergraph cuts~\cite{veldt2022hypercuts,benson2020augmented}.
%
\begin{definition}[augmented cut preserver]
	\label{def:augpres}
	Let $\mathcal{H} = (V,\mathcal{E})$ be a hypergraph with generalized cut function. The directed graph $G(\mathcal{H}) = (\hat{V},\hat{E})$ is an \emph{augmented cut preserver} for $\mathcal{H}$ if $\hat{V} = V \cup \mathcal{A}$ where $\mathcal{A}$ is an auxiliary node set and if $G(\mathcal{H})$ preserves cuts in $\mathcal{H}$ in the sense that for every $S \subseteq V$:
		\begin{equation}
			\label{eq:cutpreserve}
			\cut_\mathcal{H}(S) = \min_{U \subseteq \mathcal{A}} \cut_{G(\mathcal{H})} (S\cup U).
		\end{equation}
\end{definition}
The cut preserving property in~\eqref{eq:cutpreserve} essentially says that for any fixed $S \subseteq V$, we can arrange nodes from $\mathcal{A}$ into two sides of a cut in a way that minimizes the overall (directed) cut penalty in $G(\mathcal{H})$. There is more than one way to construct an augmented cut preserver $G(\mathcal{H})$ for a cardinality-based submodular hypergraph cut function~\cite{veldt2022hypercuts,veldt2021approximate,kohli2009robust}.
\begin{figure}
	\centering
	\includegraphics[width=.6\linewidth]{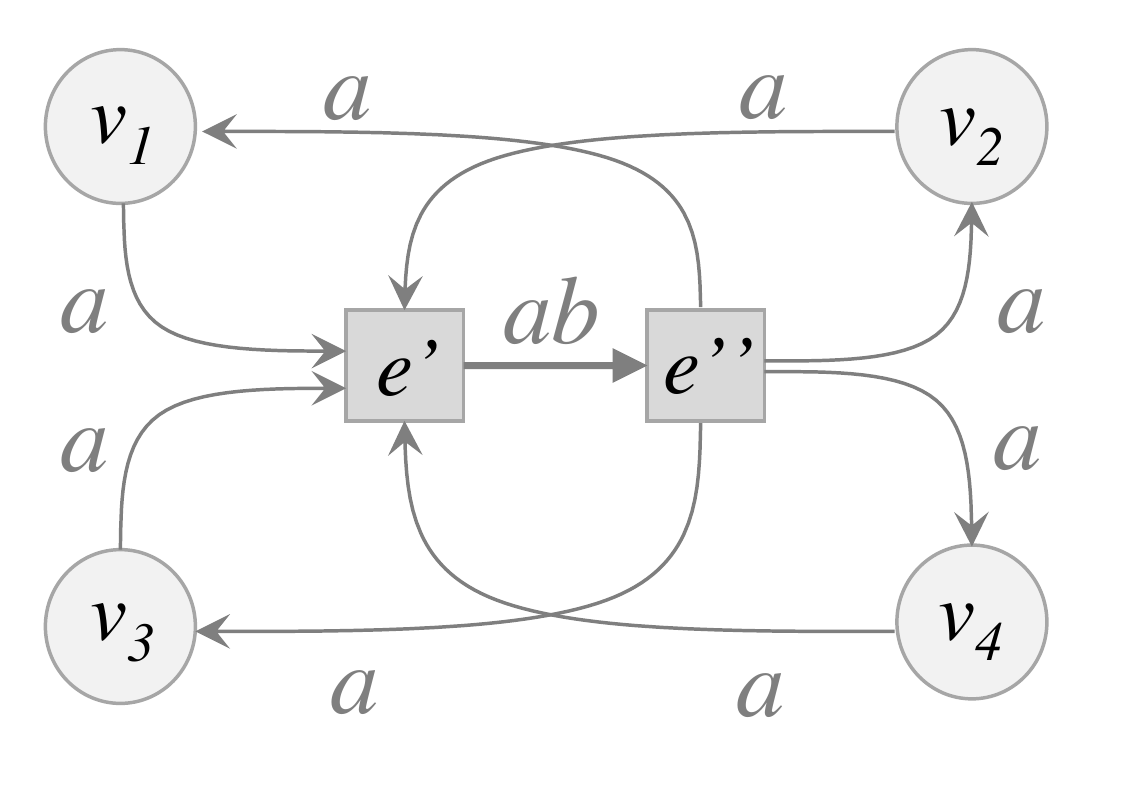}
	\caption{CB-gadget parameterized by weights $a$ and $b$ for a four-node hyperedge $e = \{v_1, v_2, v_3, v_4\}$.}
	\label{fig:gadget}
\end{figure} 
The general strategy is to replace each hyperedge $e \in \mathcal{E}$ with a small \emph{gadget} involving directed weighted edges and auxiliary nodes, in a way that models the splitting function of the hyperedge. We specifically use reductions based on symmetric cardinality-based gadgets (CB-gadgets)~\cite{veldt2022hypercuts,veldt2021approximate} (Figure~\ref{fig:gadget}). For a hyperedge $e \in \mathcal{E}$ with $k$ nodes, a CB-gadget introduces two new nodes $e'$ and $e''$, which are added to the augmented node set $\mathcal{A}$. For each $v \in e$, the gadget has two directed edges $(v,e')$ and $(e'',v)$ that are both given weight $a$, and a directed edge $(e',e'')$ with weight $ab$, where $a$ and $b$ are nonnegative parameters for the CB-gadget. 

Figure~\ref{fig:gh} is an illustration of an augmented cut preserver for a small hypergraph, where each hyperedge is replaced by a single CB-gadget and edge-weights are omitted. Each hyperedge could also be replaced by a combination of CB-gadgets with different weights, in order to model more complicated hypergraph cut functions. It is possible to model any cardinality-based submodular splitting function using a combination of $\floor{|e|/2}$ CB-gadgets, by carefully choosing parameters $(a,b)$ for each gadget~\cite{veldt2022hypercuts}. In cases where it is enough to approximately model the hypergraph cut function, one can instead use an augmented sparsifier~\cite{benson2020augmented} to approximate the cut function while using fewer CB-gadgets. 

\begin{figure*}[t]
	\centering
	\subfigure[$\mathcal{H}$\label{fig:H} ]
	{\includegraphics[width=.245\linewidth]{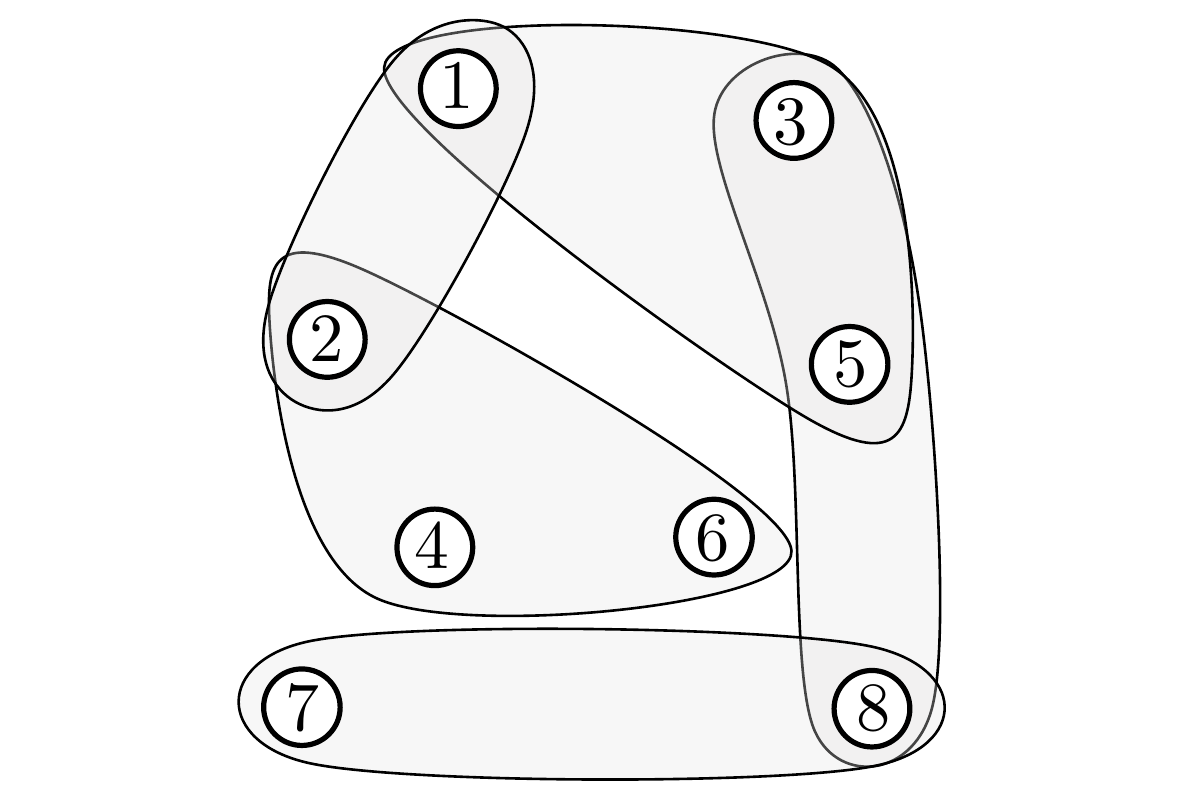}}
	\subfigure[$G(\mathcal{H})$\label{fig:gh} ]
	{\includegraphics[width=.245\linewidth]{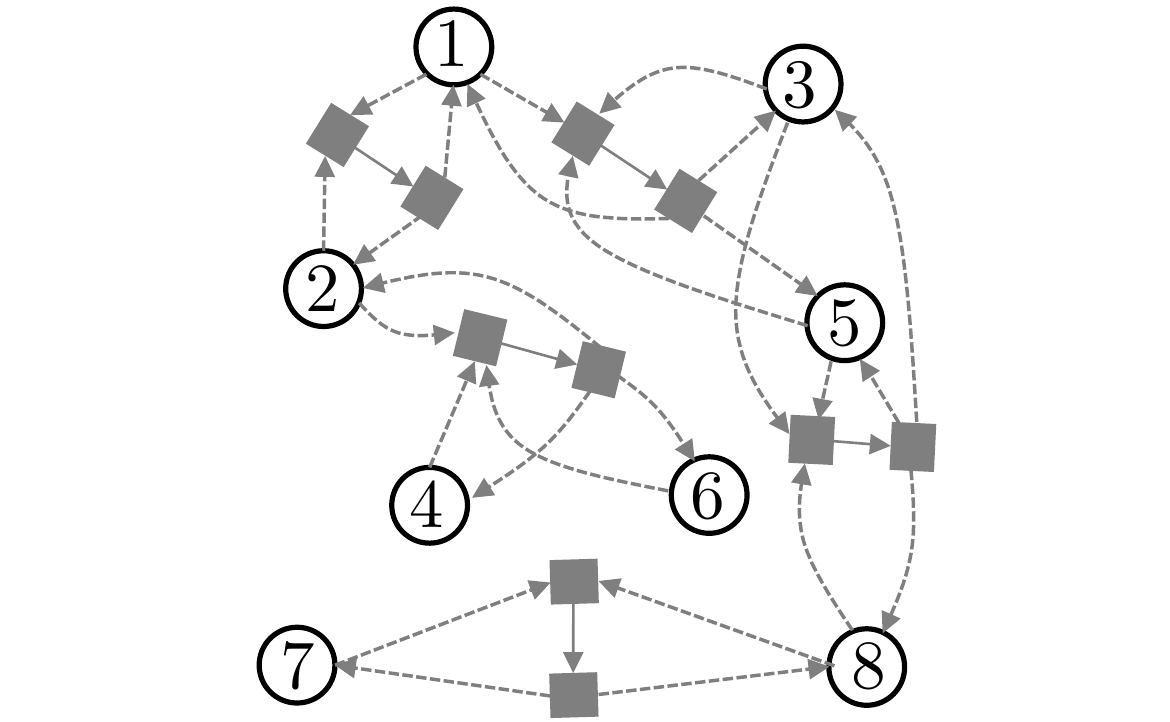}}
	\subfigure[$G(\mathcal{H},R,\alpha)$\label{fig:ghr} ]
	{\includegraphics[width=.245\linewidth]{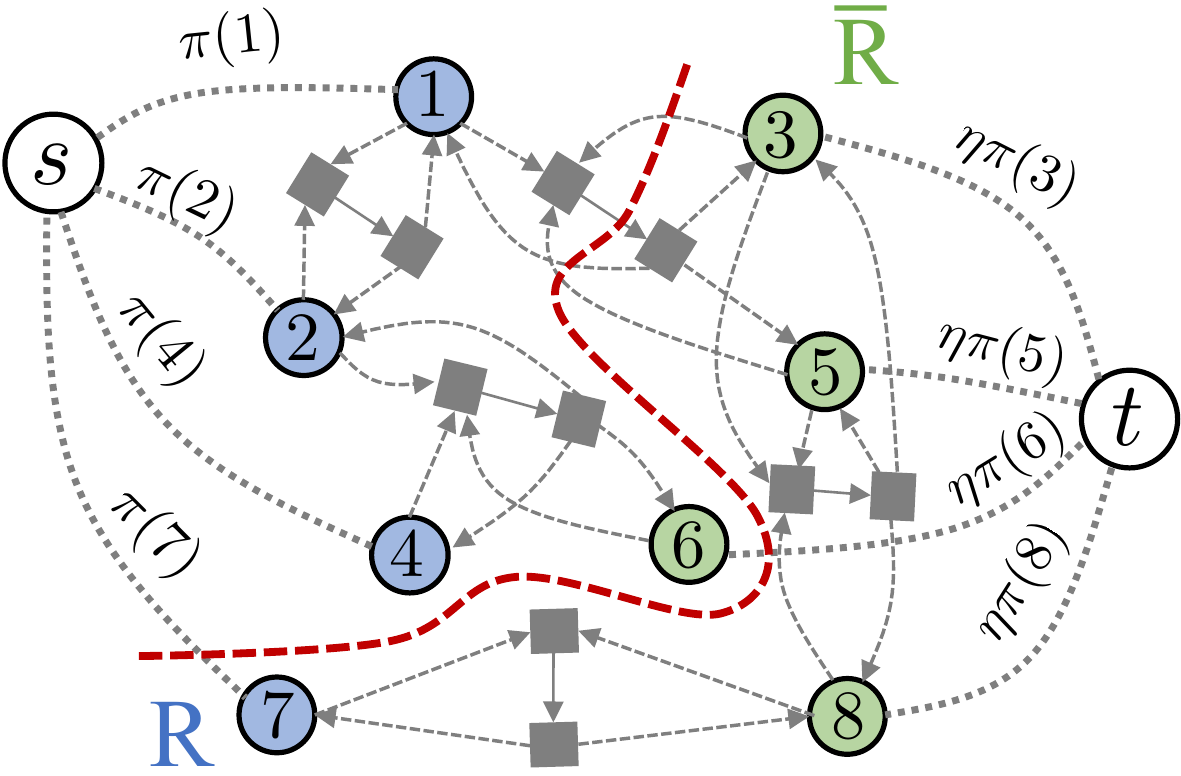}}
	\subfigure[$M_R$\label{fig:mr} ]
	{\includegraphics[width=.245\linewidth]{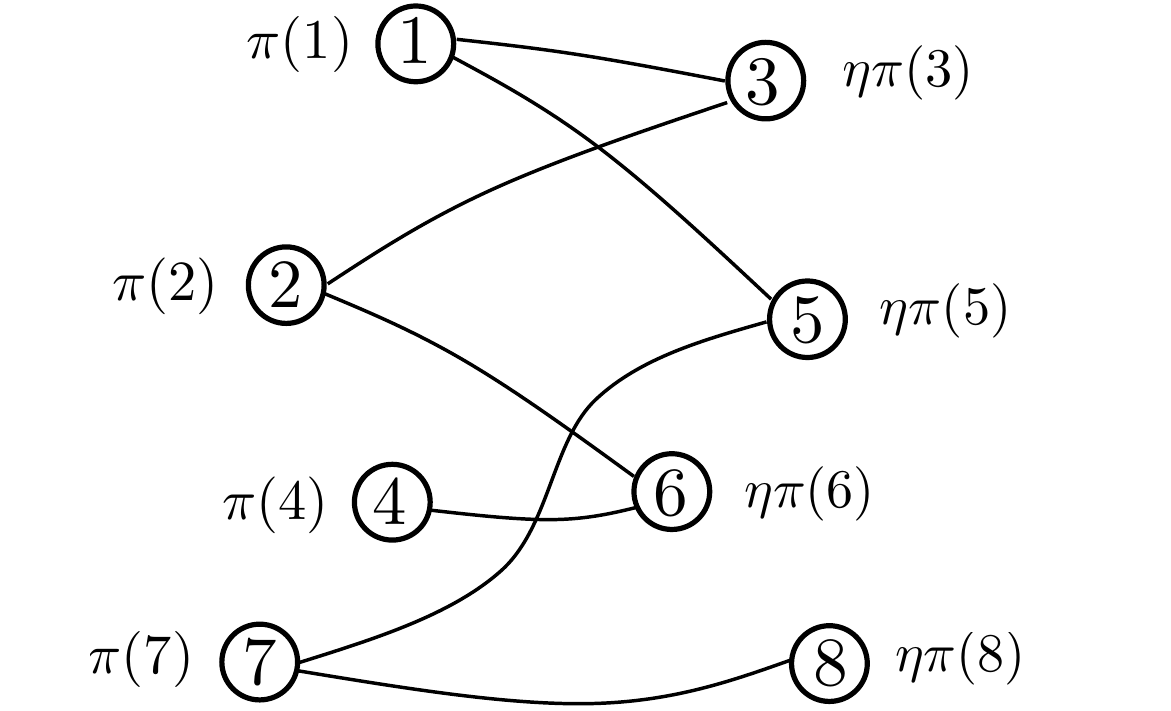}}
	\vspace{-1.0\baselineskip}
	\caption{(a) An example hypergraph $\mathcal{H} = (V,\mathcal{E})$ with 8 nodes and 5 hyperedges. (b) An \emph{augmented cut preserver} $G(\mathcal{H})$ (Def.~\ref{def:augpres}) models the generalized hypergraph cut function of $\mathcal{H}$ using the cut properties in a directed graph with an augmented node set. (c) For a given partition $\{R, \bar{R}\}$ and a parameter $\alpha$, the auxiliary graph $G(\mathcal{H},R,\alpha)$ involves source and sink nodes $s$ and $t$ and additional weighted edges. Solving a max-flow problem in this graph will either produce small cut (dashed red line) and a set of nodes $S \subseteq V$ with $\pi$-expansion less than $\alpha$, or a $\pi$-regular bipartite graph $M_R$ (pictured in (d)) that can be embedded in $G(\mathcal{H})$ with congestion $1/\alpha$. Repeatedly solving max-flow problems for many (carefully chosen) sets $R$ produces many $\pi$-regular bipartite graphs. The union of these bipartite graphs is an expander graph that can be embedded in $G(\mathcal{H})$ with a congestion that is bounded in terms of the $\pi$-expansion of the best set $S^*$ found along the way. This leads to an $O(\log n)$ approximation.}
	\label{fig:examples}
	\vspace{-.95\baselineskip} 
\end{figure*}

\subsection{Expander embeddings in directed graphs}
 Let $G = (V,E_G)$ be a directed graph with edge weight $w_G(u,v)$ for each $(u,v) \in E_G$, and let $H = (V,E_H)$ be an undirected graph on the same set of nodes with weight $w_H(i,j) \geq 0$ for edge $\{i,j\} \in E_H$. 
 For every set $S \subseteq V$, let $\mathcal{D}(S) = \{ (u,v) \in S \times \bar{S} \colon \{u,v\} \in E_H\}$ denote a set of directed pairs of nodes that share an edge in $H$. 
\begin{definition}[{embedding $H$ in $G$}]
	\label{def:embed} 
 Graph $H$ can be embedded in $G$ with congestion $\gamma$ if for each bisection $\{S, \bar{S}\}$:
	\begin{itemize}
		\item For each $(u,v) \in \mathcal{D}(S)$ there is a $u$-$v$ flow function $f^{(uv)}$ that routes $w_{H}(u,v)$ units of flow from $u$ to $v$ via edges in $G$.
		\item The flow $f = \sum_{(u,v) \in \mathcal{D}(S)} f^{(uv)}$ can be routed through $G$ with congestion at most $\gamma$, i.e., for each $(i,j) \in E_G$,
		\begin{equation}
			\label{eq:congestion}
			\sum_{(u,v) \in \mathcal{D}(S)} f_{ij}^{(uv)} \leq \gamma w_G(i,j).
		\end{equation}
	\end{itemize}
\end{definition}

\begin{lemma}
	\label{lem:expbound}
	If $H$ is embedded in $G$ with congestion $\gamma$, then for every node weight function $\pi$, $\phi_{G,\pi} \geq \frac{1}{\gamma} \phi_{H,\pi}$.
\end{lemma}
\begin{proof}
Let $S \subseteq V$ be an arbitrary set of nodes satisfying $\pi(S) \leq \pi(\bar{S})$. For each pair $(u,v) \in \mathcal{D}(S)$ there is a flow $f^{(uv)}$ over edges in $E_G$ with flow value $w_H(u,v)$, and these flows can simultaneously be routed in $E_G$ with congestion $\gamma$. All of these flows must pass through the cut edges in $E_G$, since each such pair $(u,v)$ crosses the bipartition $\{S, \bar{S}\}$. Thus, we have
\begin{align*}
	\cut_H(S) &= \sum_{(u,v) \in \mathcal{D}(S)} w_{H}(u,v) \leq \sum_{(u,v) \in \mathcal{D}(S)} \sum_{(i,j) \in \partial_GS} f_{ij}^{(uv)}\\
	& \leq \gamma \sum_{(i,j) \in \partial_GS} w_{G}(i,j) = \gamma \cut_G(S).
\end{align*}	
Therefore, $\phi_G(S) = \frac{\cut_G(S)}{\pi(S)} \geq \frac{1}{\gamma} \frac{\cut_H(S)}{\pi(S)} = \frac{1}{\gamma} \phi_H(S)$.\\
\end{proof}
The immediate implication of Lemma~\ref{lem:expbound} is that if $H$ is an expander, then the expansion of $G$ is $\Omega(\frac{1}{\gamma})$. 
Definition~\ref{def:embed} differs slightly from other notions of embeddings that are often used when approximating graph ratio cuts.
This is because of the nuances that arise when generalizing these expander embedding techniques to the hypergraph setting, especially when considering generalized hypergraph cut functions. In particular, approximating ratio cuts in an \emph{undirected} hypergraph involves considering flows in a \emph{directed} graph, and care must be taken when considering the direction of flow paths for an arbitrary set $S \subseteq V$ and a pair of nodes $\{u,v\}$ on opposite sides of a partition $\{S,\bar{S}\}$. Definition~\ref{def:embed} provides a formal notion of expander embedding that will be particularly convenient when we are dealing with directed flow problems in Section~\ref{sec:flowmatch}.

\subsection{Hypergraph expander embeddings}

Combining existing notions of hypergraph cut preservers~\cite{veldt2022hypercuts,veldt2021approximate} with Definition~\ref{def:embed} provides a new strategy for bounding expansion for generalized hypergraph cuts. 
\begin{lemma}
	\label{lem:hypbound}
	Let $\mathcal{H} = (V,\mathcal{E})$ be a hypergraph with a generalized hypergraph cut function $\cut_\mathcal{H}$ and let $G(\mathcal{H}) = (\hat{V} = V \cup \mathcal{A}, \hat{E})$ be an augmented cut preserver for $\mathcal{H}$. If graph $H = (V,E_H)$ can be embedded in $G(\mathcal{H})$ with congestion $\gamma$, then for every node weight function $\pi$, $\phi_{\mathcal{H},\pi} \geq \frac{1}{\gamma} \phi_{H,\pi}$.
\end{lemma}
\begin{proof}
Let $S \subseteq V$ be an arbitrary set of nodes satisfying $\pi(S) \leq \pi(\bar{S})$. Define $\hat{S} = S \cup U$ where $U \subseteq \mathcal{A}$ is chosen in such a way that $\cut_\mathcal{H}(S) = \cut_{G(\mathcal{H})}(\hat{S})$.
Recall that $\mathcal{D}(S)$ denotes the set of pairs $(u,v)$ where $\{u,v\} \in E_H$ with $u \in S$ and $v \in V - S$. Since $H$ is embedded in $G(\mathcal{H})$ with congestion $\gamma$, for each pair $(u,v) \in \mathcal{D}(S)$ we have routed $w_H(u,v)$ flow from $u$ to $v$. Summing up all of the flows crossing the cut gives:
\begin{align*}
	\cut_H(S) &= \sum_{(u,v) \in \mathcal{D}(S)} w_H(u,v) \leq \gamma \sum_{(i,j) \in \partial_{G(\mathcal{H})}\hat{S}} w_{G(\mathcal{H})}(i,j)\\
	& = \gamma \cut_{G(\mathcal{H})}(\hat{S}) = \gamma \cut_\mathcal{H}(S).
\end{align*}
Thus, we have $\phi_\mathcal{H}(S) = \frac{\cut_\mathcal{H}(S)}{\pi(S)} \geq \frac{1}{\gamma} \frac{\cut_H(S)}{\pi(S)} = \frac{1}{\gamma} \phi_H(S)$. 
\end{proof}
This lemma assumes a fixed node weight function $\pi$ and therefore does not directly relate conductance in $\mathcal{H}$ to conductance in $H$ (as conductance depends on node degrees, which are not necessarily preserved between $H$ and $\mathcal{H}$). However, as we shall see, this still provides an important step in lower bounding arbitrary ratio-cut objectives in $\mathcal{H}$, including conductance.


\section{Hypergraph Flow Embedding}
\label{sec:flowmatch}
To apply Lemma~\ref{lem:hypbound}, we would like to embed an expander into $G(\mathcal{H})$ with a small congestion, and then find a set $S$ whose expansion is not too far from the lower bound implied by the lemma. As a key step in this direction, in this section we will show how to embed a special type of bipartite graph into $\mathcal{H}$ whose congestion is related to a small expansion set. Figure~\ref{fig:examples} provides an illustrated overview of the flow embedding techniques in this section. Given a partition $\{R, \bar{R}\}$ of the node set $V$, we will design a procedure that takes in a parameter $\alpha >  0$ and uses a single maximum $s$-$t$ flow computation to either (1) return a set $S \subseteq V$ such that $\phi_\mathcal{H}(S,\pi) < \alpha$, or (2) produce a bipartite graph between $R$ and $\bar{R}$ that can be embedded in $G(\mathcal{H})$ with congestion $1/\alpha$. In Section~\ref{sec:mainalg}, we will show how to combine these bipartite graphs to embed an expander into $G(\mathcal{H})$.

\subsection{Maximum flows in an auxiliary graph}
Fix $\alpha > 0$ and a partition $\{R, \bar{R}\}$ satisfying $\pi(R) \leq \pi(\bar{R})$, and set $\eta = \pi(R)/\pi(\bar{R})$. We will solve a maximum $s$-$t$ flow problem on an auxiliary graph $G(\mathcal{H},R,\alpha)$ (see Figure~\ref{fig:ghr}) constructed as follows:
\begin{itemize}
	\item Construct a CB-gadget cut preserver $G(\mathcal{H})$ for $\mathcal{H}$~\cite{veldt2022hypercuts,veldt2021approximate}, and scale all edge weights by $\frac{1}{\alpha}$.
	\item Add an extra source node $s$ and a sink node $t$.
	\item For each $r \in R$, add a directed edge $(s,r)$ with weight $\pi(r)$.
	\item For each $v \in \bar{R}$, add a directed edge $(v,t)$ with weight $\eta \pi(v)$.
\end{itemize}
We will use a solution to the maximum $s$-$t$ flow problem to either find a set $S$ with small expansion (by considering the dual minimum $s$-$t$ cut problem), or find an embeddable bipartite graph between $R$ and $\bar{R}$ with congestion $\frac1{\alpha}$ (by considering a flow decomposition). This mirrors the same strategy that is used in cut-matching games for graph expansion~\cite{khandekar2009graph}, though the construction and proofs are more involved for generalized hypergraph cut functions. 



\textbf{Finding a small $\pi$-expansion set.}
The minimum $s$-$t$ cut problem in $G(\mathcal{H},R,\alpha)$ is equivalent to solving the following optimization problem over the hypergraph $\mathcal{H}$:
\begin{equation}
	\label{eq:stcut}
	\minimize_{S \subseteq V} \; \frac{1}{\alpha} \cut_\mathcal{H}(S) + \pi(R \cap \bar{S}) + \eta \pi(\bar{R} \cap S).
\end{equation}
If we set $S = \emptyset$, this corresponds to separating node $s$ from all other nodes in $G(\mathcal{H},R,\alpha)$, which is a valid $s$-$t$ cut with value $\pi(R)$. 
If we can find an $s$-$t$ cut value that is strictly smaller than $\pi(R)$ in $G(\mathcal{H},R, \alpha)$, it will provide a set $S$ with small expansion in $\mathcal{H}$.
\begin{lemma}
	\label{lem:cut}
	If the minimum $s$-$t$ cut value in ${G}(\mathcal{H},R,\alpha)$ is less than $\pi(R)$, then the set $S$ minimizing objective~\eqref{eq:stcut} satisfies $\phi_{\mathcal{H}}(S,\pi) < \alpha$.
\end{lemma}
\begin{proof}
The minimum $s$-$t$ cut set in $G(\mathcal{H},R,\alpha)$ is some set of nodes $\{s\} \cup S \cup \mathcal{A}_S$ where $S \subseteq V$ is a set of nodes from the original hypergraph $\mathcal{H}$ and $\mathcal{A}_S$ is a subset of the auxiliary nodes from the cut preserver $G(\mathcal{H})$, designed so that the cut function in ${G}(\mathcal{H})$ matches the cut function in $\mathcal{H}$. If the minimum $s$-$t$ cut value in $G(\mathcal{H},R,\alpha)$ is less than $\pi(R)$, this means that for the set $S$:
\begin{align*}
	&\frac{1}{\alpha} \cut_\mathcal{H}(S) + \pi(R \cap \bar{S}) + \eta\pi(\bar{R} \cap S) < \pi(R) \\
	\implies & \cut_\mathcal{H}(S) + \alpha \pi(R \cap \bar{S}) + \alpha \eta\pi(\bar{R} \cap S) < \alpha\pi(R)\\
	\implies &	\cut_\mathcal{H}(S) < \alpha \pi(R \cap S) - \alpha \eta\pi(\bar{R} \cap S)\\
	\implies &\frac{\cut_\mathcal{H}(S)}{\pi(R \cap S) - \eta\pi(\bar{R} \cap S)} < \alpha.
\end{align*}
The result will hold as long as we can show that $\min\{\pi(S), \pi(\bar{S})\} \geq \pi(R \cap S) - \eta\pi(\bar{R} \cap S)$. First note that $\pi(S) \geq \pi(R \cap S) \geq \pi(R \cap S) - \eta \pi(\bar{R} \cap S)$. Additionally, we have
\begin{align*}
	\pi(R \cap S) - \eta\pi(\bar{R} \cap S) &= \pi(R \cap S) + \eta\pi(\bar{R} \cap \bar{S}) - \eta \pi(\bar{R})\\  &= \eta\pi(\bar{R} \cap \bar{S}) - \pi(R \cap \bar{S}) \leq \pi(\bar{S}).
\end{align*}
\end{proof}
Max-flow subroutines for solving an objective closely related to~\eqref{eq:stcut} have already been used to minimize localized ratio cut objectives in hypergraphs~\cite{veldt2020hyperlocal}. This generalizes earlier work on localized ratio cuts in graphs~\cite{Andersen:2008:AIG:1347082.1347154,orecchia2014flow,veldt16simple,Veldt2019flow}. The present work differs in that we will use max-flow solutions to approximate \emph{global} ratio cut objectives. A key step in doing so is showing how to embed a bipartite graph into $G(\mathcal{H})$ when Lemma~\ref{lem:cut} does not apply.

\textbf{Embedding a bipartite graph.}
\label{sec:embedbip}
If the min $s$-$t$ cut in $G(\mathcal{H},R,\alpha)$ is $\pi(R)$, then the max $s$-$t$ flow solution in $G(\mathcal{H},R,\alpha)$ saturates every edge touching $s$ or $t$. We can then define a bipartite graph $M_R$ between $R$ and $\bar{R}$ that can be embedded in $G(\mathcal{H})$ with congestion $\frac{1}{\alpha}$. 
Let $f$ be a maximum $s$-$t$ flow on $G(\mathcal{H},R,\alpha)$ with value $|f| = \pi(R)$. 
Letting $n = |V|$, define a matrix $\mM_R \in \mathbb{R}^{n \times n}$ that is initialized to be the all zeros matrix. Using Lemma~\ref{lem:flowdeomp},
we can decompose the flow $f$ into $f = \sum_{i = 1}^\ell f_i$, where for each $i$ the edges that have a positive flow in $f_i$ form either a cycle or a simple $s$-$t$ path. For our purposes we can ignore the cycles and focus on the $s$-$t$ paths. An $s$-$t$ flow path $f_i$ always starts by sending $|f_i|$ units of flow from $s$ to some node $r \in R$, and eventually ends by sending the same amount of flow through an edge $(v,t)$ where $v \in \bar{R}$. For each such flow path, we perform the update $\mM_R(r,v) \leftarrow \mM_R(r,v) + |f_i|$. After iterating through all $\ell$ flow paths, we define $M_R$ to be the bipartite graph whose adjacency matrix is $\mM_R$ (see Fig.~\ref{fig:mr}).
The construction of $G(\mathcal{H},R,\alpha)$ and the fact that $f$ saturates all edges touching $s$ and $t$ implies the following degree properties for $M_R$:
\begin{itemize}
	\item For each $r \in R$, the weighted degree of $r$ in $M_R$ is $\pi(r)$.
	\item For each $u \in \bar{R}$, the weighted degree of $u$ in $M_R$ is $\eta \pi(u)$.
\end{itemize}
Following previous terminology~\cite{orecchia2022practical}, we refer to a bipartite graph satisfying these properties as a $\pi$-regular bipartite graph. To relate this to previous cut-\emph{matching} games for graph expansion, note that if $|R| = |\bar{R}|$ and $\pi(u) = 1$ for all $u \in V$, then $M_R$ will be a fractional \emph{matching} and $\mM_R$ will be a doubly stochastic matrix. 
It is also worth noting that if $\pi(R) = \pi(\bar{R})$, then the weighted degree of every node $u \in V$ in $M_R$ will be equal to $\pi(u)$, leading to a correspondence between the conductance in $M_R$ and the $\pi$-expansion of $\mathcal{H}$. In practice, $\pi(R)$ will not always be exactly equal to $\pi(\bar{R})$, but we will tend to choose partitions that are roughly balanced in size.

\textbf{Proving a bound on congestion.}
When $|f| = \pi(R)$, the maximum $s$-$t$ flow in $G(\mathcal{H},R,\alpha)$ can be viewed as a directed multicommodity flow that is routed through a $\frac1{\alpha}$-scaled copy of $G(\mathcal{H})$. However, there is one subtle issue we must overcome in order to confirm that $M_R$ can be embedded in $G(\mathcal{H})$ with congestion $\frac1{\alpha}$. Definition~\ref{def:embed} requires that for every bisection $\{S, \bar{S}\}$, there must be a way to route $\mM_R(u,v)$ units of flow from $u$ to $v$ if $u \in S$ and $v \in \bar{S}$. However, $f$ only routes flow from $R$ to $\bar{R}$ in the \emph{directed} graph $G(\mathcal{H},R,\alpha)$. This issue does not arise in cut-matching games for undirected graphs, as each undirected edge $\{u,v\}$ can be viewed as a pair of directed edges $(u,v)$ and $(v,u)$, making it easy to send flow in two directions simultaneously. However, in the directed graph $G(\mathcal{H},R,\alpha)$, it is possible that for a given bipartition $\{S, \bar{S}\}$, the flow $f$ will send flow from a node $r \in R \cap \bar{S}$ to a node $u \in \bar{R} \cap S$, which does not directly satisfy the requirement in Definition~\ref{def:embed}. The following lemma confirms that we can overcome this. Its proof relies on carefully considering the edge structure in $G(\mathcal{H})$ and showing how to use an implicit flow-reversing procedure when necessary in order to satisfy Definition~\ref{def:embed}. A proof is included in the appendix.
\begin{lemma}
	\label{lem:match}
	If $|f| = \pi(R)$, the graph $M_R$ can be embedded in $G(\mathcal{H})$ with congestion $\frac1{\alpha}$ in the sense of Definition~\ref{def:embed}.
\end{lemma}

\subsection{The flow-embedding algorithm}
We combine Lemmas~\ref{lem:cut} and~\ref{lem:match} into a method \textsc{HyperCutOrEmbed} (Algorithm~\ref{alg:flowmatch}) for obtaining both a good cut and an embeddable bipartite graph for any input partition $\{R,\bar{R}\}$ with $\pi(R) \leq \pi(\bar{R})$.
We assume that the hypergraph $\mathcal{H}$ is connected, 
and that the hypergraph weights are scaled so that there is a minimum penalty of 1 when cutting any hyperedge. This implies a lower bound of $2/\pi(V)$ on the minimum $\pi$-expansion, which is achieved if there is a set $S$ with $\pi(S) = \pi(V)/2$ and $\cut_\mathcal{H}(S) = 1$. We also assume that node weights are scaled so that $\pi(v) \geq 1$ for each $v \in V$. \textsc{HyperCutOrEmbed} repeatedly solves maximum $s$-$t$ flow problems to find either a bipartite graph or a cut with bounded $\pi$-expansion. The algorithm uses black-box subroutines for maximum $s$-$t$ flows and minimum $s$-$t$ cuts, and a procedure \textsc{FlowEmbed} that decomposes a flow into a $\pi$-regular bipartite graph as outlined in Section~\ref{sec:embedbip}.
\begin{theorem}
	\label{thm:cutmatch}
	\textsc{HyperCutOrEmbed} returns a $\pi$-regular bipartite graph $M_R$ that can be embedded in $G(\mathcal{H})$ with congestion $1/\alpha$ and a set $S$ with $\phi_\mathcal{H}(S,\pi) < 2\alpha$ for some $\alpha \geq 2/\pi(V)$. The algorithm terminates in $O(\log |\mathcal{E}| + \log U + \log \pi(V))$ iterations where $U$ is the maximum hyperedge cut penalty.
\end{theorem}
\begin{proof}
In each iteration the algorithm computes a maximum $s$-$t$ flow in $G(\mathcal{H},R,\alpha)$. By Lemma~\ref{lem:match}, if the flow value is $\pi(R)$ then we can return a $\pi$-regular bipartite graph $M_R$ that can be embedded with congestion $1/\alpha$. Otherwise, Lemma~\ref{lem:cut} guarantees that the minimum $s$-$t$ cut set has expansion less than $\alpha$. The algorithm is guaranteed to return a bipartite graph on the first iteration, because it is impossible to find a set with $\pi$-expansion less than the minimum value $2/\pi(V)$. If the algorithm finds a set $S$ with $\phi_\mathcal{H}(S,\pi) < \alpha$ and then terminates, this means that in the previous iteration it found a bipartite graph $M_R$ that can be embedded with congestion $\alpha/2$. 

For every $\alpha > U|\mathcal{E}|$, finding a maximum $s$-$t$ flow in $G(\mathcal{H},R,\alpha)$ is guaranteed to return a cut set $S$ with $\phi_\mathcal{H}(S,\pi) < \alpha$. This is because $\cut_\mathcal{H}(R) \leq  {U |\mathcal{E}|}$, so the node set $\{s \cup R\}$ is an $s$-$t$ cut set in the auxiliary graph with cut value
\begin{equation*} 
	\frac{\cut_\mathcal{H}(R)}{\alpha} \leq \frac{U|\mathcal{E}|}{\alpha} < 1 \leq \pi(R).
\end{equation*}
Since \textsc{HyperCutOrEmbed} starts at $\alpha = 2/\pi(V)$ and doubles $\alpha$ at every iteration, it will take at most $O(\log U |\mathcal{E}| \pi(V))$ iterations before returning a cut set instead of a matching.
\end{proof}
For the unweighted all-or-nothing cut function, \textsc{HyperCutOrEmbed} terminates in $O(\log |V| + \log |\mathcal{E}|)$ rounds when applied to expansion or conductance. In the more general setting, the iteration bound
only depends logarithmically on the hypergraph size and weights.
\begin{algorithm}[tb]
	\caption{ $\textsc{HyperCutOrEmbed}(\mathcal{H}, R,\pi)$}
	\begin{algorithmic}[5]
		\State{\bfseries Input:} $\mathcal{H} = (V, \mathcal{E})$, node weights $\pi$, bisection $\{R, \bar{R}\}$
		\State {\bfseries Output:} $M_R$ with congestion $1/\alpha$; $S$ with $\phi_\mathcal{H}(S,\pi) \leq 2\alpha$
		\State Set $\alpha = 2/\pi(V)$, \textsc{NoCutFound} = true
		\While{ \textsc{NoCutFound}}
		\State $f$ = $\textsc{MaxSTflow}(G(\mathcal{H},R,\alpha))$.
		\If{$|f| = \pi(R)$}
		\State $M_R \leftarrow \textsc{FlowEmbed}(G(\mathcal{H}),R,f)$; $\alpha \leftarrow 2\alpha$
		\Else
		\State $S = \textsc{MinSTcut}(G(\mathcal{H},R,\alpha), f)$
		\State  \textsc{NoCutFound}  = false
		\EndIf
		\EndWhile
		\State Return $M_R$, $S$, $\alpha$
	\end{algorithmic}
	\label{alg:flowmatch}
	\vspace{-0.2\baselineskip} 
\end{algorithm}

\section{Hypergraph Ratio Cut Algorithms}
\label{sec:mainalg}
\textsc{HyperCutOrEmbed} finds a node set whose $\pi$-expansion is related to a graph that can be embedded in $\mathcal{H}$. Applying Lemma~\ref{lem:hypbound} directly does not imply a useful lower bound or approximation algorithm for $\pi$-expansion in $\mathcal{H}$, as the bipartite graph itself does not have a large $\pi$-expansion. In this section we show how to combine \textsc{HyperCutOrEmbed} with existing strategies for building an expander graph in order to design an approximation algorithm for hypergraph $\pi$-expansion. 

\subsection{Expander building subroutines}
The standard cut-matching procedure for expansion in an undirected graph $G = (V,E)$ can be described as a two-player game between a cut player and a matching player. At the start of iteration $i$, the cut player produces a bisection $\{R_i, \bar{R}_i\}$, and the matching player produces a fractional perfect matching $M_i$ between $R_i$ and $\bar{R}_i$, encoded by a doubly stochastic matrix $\mM_i \in [0,1]^{|V| \times |V|}$. 
After $t$ iterations, the union of matchings defines a graph $H_t = \bigcup_{j = 1}^t M_j$ with adjacency matrix $\mA_t = \sum_{j = 1}^t \mM_j$. The goal of the cut player is to choose bisections in a way that minimizes the number of rounds it takes before $H_t$ is an expander, while the goal of the matching player is to choose matchings that maximize the number of rounds. Algorithm~\ref{alg:expbuild} is a generic outline of this procedure, where \textsc{FindBisection} and \textsc{FractionalMatch} are the cut player subroutine and the matching player subroutine respectively. 
Khandekar et al.~\cite{khandekar2009graph} provided a strategy for the cut player which, for any matching player subroutine \textsc{FractionalMatch}, will force $H_t$ to have expansion at least $1/2$ for some $t = O(\log^2 n)$ with high probability. This was used to show an $O(\log^2 n)$ approximation for graph expansion.

\begin{algorithm}[t]
	\caption{Generic cut-matching game for building an expander}
	\begin{algorithmic}
		\State $M_0 = (V,\emptyset)$
		\For{$i = 1$ to $t$}
		\State $R_i = \textsc{FindBisection}(M_0, M_1, M_2, \hdots , M_{i-1})$
		\State $M_i = \textsc{FractionalMatch}(\{R_i, \bar{R}_i\})$
		\State $H_i = \bigcup_{j = 1}^i M_j$
		\EndFor
	\end{algorithmic}
	\label{alg:expbuild}
\end{algorithm}

Since the work of Khandekar et al.~\cite{khandekar2009graph}, the cut-matching framework has been applied, improved, and generalized in many different ways~\cite{orecchia2008partitioning,sherman2009breaking,louis2010cut,orecchia2011fast,orecchia2022practical}. Although previous work on cut-matching games has largely focused on expansion node weights ($\mu(v) = 1$), Orecchia et al.~\cite{orecchia2022practical} recently introduced a setting where the cut player produces a set $R_i$ satisfying $\pi(R_i) \leq \pi(\bar{R}_i)$ at each iteration, and the matching player produces a $\pi$-regular bipartite graph $M_i$ on $\{R_i, \bar{R}_i\}$. Taking the union of all bipartite graphs up through the $i$th iteration produces a graph $H_i = (V,E_{H_i})$. The cut player wishes to make $H_i$ a $\pi$-expander for small $i$. Lemma~\ref{lem:thm7} summarizes a cut player strategy that applies to this setting and leads to an $O(\log n)$-approximation for graph ratio cut objectives.
\begin{lemma}[Theorem 7 in~\cite{orecchia2022practical}]
	\label{lem:thm7}
	There exists a cut player strategy such that, for any matching player strategy, $H_t$ satisfies $\phi_{H_t,\pi} = \Omega (\log n)$ with high probability for some $t = O(\log^2 n)$. In round $i$, the cut player can compute $R_i$ in time $O(|E_{H_i}| \cdot \text{polylog}(\pi(V)))$.
\end{lemma}
This strategy is based on using a heat-kernel random walk in the graph $H_t$. The approach was first used for expansion weights ($\mu(v) = 1$) by Orecchia et al.~\cite{orecchia2008partitioning} and considered in more depth in Orecchia's PhD thesis~\cite{orecchia2011fast}. Lemma~\ref{lem:thm7} generalizes this to general node weights and was presented in recent work on overlapping graph clustering~\cite{orecchia2022practical}. More technical details on cut-matching games with general node weights were presented recently by an overlapping group of authors~\cite{ameranis2023efficient}. We refer to the cut player strategy of Lemma~\ref{lem:thm7} as $\textsc{HeatKernelPartition}$. In the first iteration, $H_0$ is empty, so the cut player starts with any balanced bipartition $\{R, \bar{R}\}$.

\subsection{The approximation algorithm}
The notion of expander embedding that we have introduced involves embedding an expander {graph} into a hypergraph $\mathcal{H}$. Therefore, we can use existing cut player strategies to build a $\pi$-expander, which we embed into $\mathcal{H}$ in the sense of Definition~\ref{def:embed} using algorithm \textsc{HyperCutOrEmbed}. Our cut-matching approximation algorithm for generalized hypergraph cut expansion (Algorithm~\ref{alg:cutmatch}) is obtained by replacing the generic fractional matching subroutine in Algorithm~\ref{alg:expbuild} with \textsc{HyperCutOrEmbed}, and using \textsc{HeatKernelPartition} for the cut player step.
\begin{algorithm}[t]
	\caption{$O(\log n)$ approximation for hypergraph $\pi$-expansion}
	\begin{algorithmic}
		\State{\bfseries Input:} $\mathcal{H} = (V, \mathcal{E})$, generalized $\cut_\mathcal{H}$, bisection $\{R, \bar{R}\}$
		\State {\bfseries Output:} Set $S$ with small $\pi$-expansion
		\State $M_0 = \emptyset$
		\For{$i = 1$ to $t$}
		\State $R_i = \textsc{HeatKernelPartition}(M_0,M_1, M_2, \hdots , M_{i-1})$
		\State $M_i, S_i, \alpha_i = \textsc{HyperCutOrEmbed}(\mathcal{H},R_i)$
		\EndFor
		\State Return $S^* = \argmin_{i = 1,2, \hdots t} \phi_\mathcal{H}(S_i,\pi)$
	\end{algorithmic}
	\label{alg:cutmatch}
\end{algorithm}
\begin{theorem}
	For some $t = O(\log^2 n)$, Algorithm~\ref{alg:cutmatch} is an $O(\log n)$-approximation algorithm for hypergraph $\pi$-expansion where $\cut_\mathcal{H}$ is any submodular cardinality-based hypergraph cut function.
\end{theorem}
\begin{proof}
	At iteration $i$, \textsc{HyperCutOrEmbed} produces a $\pi$-regular bipartite graph $M_i$ that can be embedded in $G(\mathcal{H})$ with expansion $\frac{1}{\alpha_i}$ and a set $S_i$ with expansion $\phi_{\mathcal{H}}(S_i,\pi) \leq 2\alpha_i$.
	Define $i^* = \argmin_i \alpha_i$ and note that $\phi_\mathcal{H}(S^*,\pi) \leq \phi_\mathcal{H}(S_{i^*},\pi)$.
	The union of bipartite graphs $H_t$ can be embedded in $G(\mathcal{H})$ with congestion $\sum_{i = 1}^t \frac{1}{\alpha_i} \leq \frac{t}{\alpha_{i^*}}$. From Lemma~\ref{lem:thm7}, we have $t \leq c_1 \log^2 n$ and $\phi_{H_t,\pi} \geq c_2 \log n$ with high probability, where $c_1$ and $c_2$ are positive constants. Combining this with Lemma~\ref{lem:hypbound} we have
	\begin{equation*}
		\frac{\alpha_i^*}{\log n} \leq \frac{c_1 \alpha_i^*}{t}\log n \leq \frac{c_1 \alpha_i^*}{c_2t} \phi_{H_t,\pi} \leq \frac{c_1}{c_2} \phi_{\mathcal{H},\pi}. 
	\end{equation*}
So $\phi_{\mathcal{H},\pi} = \Omega\left({\alpha_i^*}/{\log n}\right)$, and the algorithm returns a set $S^*$ with $\phi_\mathcal{H}(S^*,\pi) \leq 2\alpha_i^*$, proving the $O(\log n)$ approximation guarantee.
%
%
\end{proof}

\subsection{Runtime Analysis}
For our runtime analysis, we assume the hypergraph is connected and that all hyperedge weights and node weights are scaled to be integers. Let $n = |V|$, $m = |\mathcal{E}|$, and $\mu = \sum_{e \in \mathcal{E}} |e|$. In order to focus on the main terms in the runtime, will use $\tilde{O}$ notation to hide logarithmic factors of $m$ and $n$. For our analysis we also assume that the maximum edge penalty $U$ and sum of node weights are small enough that $O(\log U)$ and $O(\log \pi(V))$ are both $\tilde{O}(1)$, and hence the maximum number of iterations of~\textsc{HyperCutOrEmbed} is $\tilde{O}(1)$. When choosing a node weight function corresponding to conductance we have $\log \pi(V) = O(\log mn)$, and choosing the $\pi$ corresponding to standard expansion we have $\log \pi(V) = \log n$. 

In order to speed up the runtime for Algorithm~\ref{alg:cutmatch}, we can apply existing sparsification techniques for hypergraph-to-graph reduction~\cite{veldt2021approximate}. This allows us to model the generalized cut function of $\mathcal{H}$ to within a factor of $(1+\varepsilon)$ for a small constant $\varepsilon > 0$ with an augmented graph $G(\mathcal{H})$ with $N = O(n + \sum_{e \in \mathcal{E}} \log |e|) = \tilde{O}(n + m)$ nodes and $M = O(\sum_{e \in \mathcal{E}} |e| \log |e|) = \tilde{O}(\mu)$ edges. 
Constructing this graph takes $O(M)$ time.  For $\alpha > 0$, the graph $G(\mathcal{H}, R, \alpha)$ also has $O(N)$ nodes and $O(M)$ edges. 

The flow decomposition step in \textsc{HyperCutOrEmbed} can be accomplished in $O(M \log N) = \tilde{O}(\mu)$ time using dynamic trees~\cite{sleator1983data}; note that for this step we do not need to explicitly return the entire flow decomposition but simply must identify the endpoints in each directed path for the bipartite graph we are embedding. Lemma~\ref{lem:thm7} indicates that the total time spent on the cut player strategy will be $\tilde{O}(|E_{H_t}|)$, which we can bound above by $\tilde{O}(\mu)$. To see why, observe that the number of edges we add to the bipartite graph constructed by \textsc{FlowEmbed} will be bounded above by the number of different directed flow paths, which by Lemma~\ref{lem:flowdeomp} will be bounded above by $O(M) = O(\mu)$. Combining the edges from all $O(\log^2 n)$ bipartite graphs shows $\tilde{O}(|E_{H_t}|) = \tilde{O}(\mu)$. The overall runtime of our algorithm is dominated by the time it takes to solve a maximum $s$-$t$ flow in $G(\mathcal{H},R,\alpha)$. This overall runtime is $\tilde{O}(\mu + (n+m)^{3/2})$ if using the algorithm of van den Brand et al.~\cite{brand2021minimum}. The recent algorithm of Chen et al.~\cite{chen2022maximum} brings the runtime down to $\tilde{O}(\mu^{1+o(1)})$, nearly linear in terms of the hypergraph size $\mu$. 

For comparison, the existing LP relaxation~\cite{kapralov2020towards} and SDP relaxation~\cite{louis2016approximation}, which only apply to all-or-nothing hypergraph cuts, both involve $\Omega(n^2 + m)$ variables and $\Omega(n^3 + \sum_{e \in E} |e|^2)$ linear constraints. For many choices of $m$, $n$, and $\mu$, simply {writing down} these relaxations has a slower runtime than running our method, while actually solving them is of course far worse. When written in the form $\min_{\textbf{A}\textbf{x} = \textbf{b}} \textbf{c}^T\vx$, the LP has $\Omega(n^3 + \sum_{e \in E} |e|^2)$ constraints and variables. For problems of this size, even recent breakthrough theoretical results in LP solvers~\cite{jiang2020faster,cohen2021solving} lead to runtimes significantly worse than $\Omega(n^6 + n^3\sum_{e} |e|^2 + \mu^2)$.


\subsection{Practical improvements} 
\label{sec:practical}
We incorporate a number of practical updates to simplify our algorithm's implementation and improve approximation guarantees.

\textbf{Practical flow algorithms.}
First of all, our implementation uses the push-relabel max-flow algorithm~\cite{cherkassky1997implementing}, which has a worse (but still fast) theoretical runtime and comes with various heuristics that make it very fast in practice. For the flow decomposition, we use a standard decomposition technique (see Theorem 3.5 in~\cite{ahuja1988network}) that does not require dynamic trees. This is again theoretically slower but still fast and is simpler to work with in practice. 

\textbf{Update to \textsc{HyperFlowEmbed}.} 
We also use an altered version of \textsc{HyperFlowEmbed} that returns a bipartite graph $M_R$ that can be embedded with congestion $1/\alpha$ as well as a set $S$ that has $\pi$-expansion equal to $\alpha$, rather than just a set $S$ with $\phi_\mathcal{H}(S,\pi) < 2\alpha$.  This improves the approximation by a factor of 2, which does not change the theoretical worst case $O(\log n)$ approximation but can make a substantial difference in practice.
To accomplish this, we set $\alpha = \phi_\mathcal{H}(R,\pi)$ in the first iteration and solve a maximum $s$-$t$ flow problem on $G(\mathcal{H},R,\alpha)$ to search for a set $S$ with $\pi$ expansion better than $\phi_\mathcal{H}(R,\pi)$. In each iteration, we update $\alpha$ to equal the $\pi$-expansion of the improved set found in the previous iteration, until no more improvement is found. This iterative refinement approach is standard and typically used in practice by related ratio cut improvement algorithms~\cite{LangRao2004,veldt16simple,veldt2020hyperlocal,orecchia2022practical}. Although performing a bisection method over $\alpha$ leads to better theoretical runtimes, in practice it typically takes only a few iterations of cut improvement before the iterative refinement procedure converges. Thus, this is often faster in practice in addition to improving the approximation guarantee by a factor 2. 

\textbf{Improved lower bounds and a posteriori approximations.}
Finally, we establish more precise lower bounds on the approximation guarantee satisfied by the algorithm in each iteration. This allows us to obtain improved a posteriori approximation guarantees. For node weight function $\pi$, let $\mD_\pi$ be the diagonal matrix where $\mD_\pi(i,i) = \pi(i)$ and define $\mathcal{L}_t = \mD_\pi^{-1/2}\mL_t \mD_\pi^{-1/2}$. 
Let $\lambda_2(\mathcal{L}_t)$ be the 2nd smallest eigenvalue of this matrix. 
%
The following theorem presents a precise and easy-to-compute lower bound on the approximation factor achieved by our approximation algorithm at each iteration. For this result, we use the updated version of $\textsc{HyperFlowEmbed}$  that returns a set $S$ with $\phi_{\mathcal{H}}(S,\pi) = \alpha$.

\begin{theorem}
	\label{thm:bounds}
	After $t$ iterations, Algorithm~\ref{alg:cutmatch} will return a set $S_t^*$ satisfying $\phi_\mathcal{H}(S_t^*,\pi) \leq \rho_t \phi_\mathcal{H,\pi}$, where $\rho_t = \frac{2 \gamma_t \phi_\mathcal{H}(S_t^*,\pi) }{\lambda_2(\mathcal{L}_t)} \leq \frac{2 t }{\lambda_2(\mathcal{L}_t)}$ and $\gamma_t = \sum_{i = 1}^t \frac{1}{\alpha_i}$.
\end{theorem}
\begin{proof}
The result follows by expressing the eigenvalue as a Raleigh quotient and using standard spectral graph theory techniques to first show that
\begin{equation}
	\label{eq:lambda2}
	\lambda_2(\mathcal{L}_t) = \min_{\substack{\vx \in \mathbb{R} \\ \vx^T \mD_{\pi} \ve = 0}} \frac{ \vx^T \mL_t \vx}{\vx^T \mD_\pi \vx} \leq 2 \cdot \min_{S \subseteq V} \phi_{H_t}(S, \pi) = 2  \phi_{H_t,\pi}.
\end{equation}
This lower bound on the $\pi$-expansion in $H_t$ can be shown by considering an arbitrary set $S \subseteq V$ and defining a vector $\vx$ by
\begin{equation*}
	\vx(i) = \begin{cases}
		\frac{1}{\pi(S)} & \text{ if $i \in S$} \\
		-\frac{1}{\pi(\bar{S})} & \text{ if $i \in \bar{S}$}.
	\end{cases}
\end{equation*}
Observe that this vector satisfies $\vx^T \mD_{\pi} \ve = 0$, and 
\begin{equation*}
	\frac{ \vx^T \mL_t \vx}{\vx^T \mD_\pi \vx} = \frac{\cut_{H_t}(S)}{\pi(S)} + \frac{\cut_{H_t}(S)}{\pi(\bar{S})} \leq 2 \phi_{H_t}(S,\pi).
\end{equation*}
Combining the bound in~\eqref{eq:lambda2} with the fact that $H_t$ can be embedded in $\mathcal{H}$ with congestion $\gamma_t$ gives the following lower bound on $\phi_{\mathcal{H},\pi}$:
\begin{equation*}
 \frac{1}{2\gamma_t}	\lambda_2(\mathcal{L}_t) \leq \frac{1}{\gamma_t} \phi_{H_t,\pi}  \leq \phi_\mathcal{H,\pi}.
\end{equation*}
Dividing $\phi_\mathcal{H}(S_t^*,\pi)$ by this lower bound proves the approximation.

\end{proof}
This theorem suggests another alternative for the cut player strategy: choose a partition $\{R, \bar{R}\}$ by thresholding entries in the second smallest eigenvector of $\mathcal{L}_t$. Embedding a $\pi$-regular bipartite graph across this partition increases the value of $\lambda_2(\mathcal{L}_t)$ in subsequent iterations. There exist extreme cases where greedily choosing a bipartition based on this eigenvector makes slow progress (see discussion on page 35 of~\cite{orecchia2011fast}). However, this tends to produce good results in practice, and we can use Theorem~\ref{thm:bounds} to compute concrete lower bounds on $\pi$-expansion using this strategy.

	\section{Experiments}
We implement our hypergraph cut matching algorithm (HCM) in Julia using all of the practical improvements from Section~\ref{sec:practical}. We compare it against other hypergraph ratio-cut algorithms in minimizing global ratio cut objectives on hypergraphs from various application domains. All experiments were run on an Macbook Air with an Apple M1 Chip and 16GB of RAM. All data and code is provided at~\url{https://github.com/nveldt/HyperCutMatch}. The appendix provides additional details on implementations and datasets. 

\textbf{Datasets.}
\emph{Trivago} encodes sets of vacation rentals that a single user clicks on during a browsing session on \texttt{Trivago.com}~\cite{chodrow2021generative}. \emph{Mathoverflow} encodes groups of posts on \texttt{mathoverflow.com} that a single user answers~\cite{veldt2020hyperlocal}. \emph{Amazon9}~\cite{veldt2020parameterized} encodes Amazon retail products (from the 9 smallest product categories of a larger dataset~\cite{ni2019justifying}) that are reviewed by the same user. \emph{TripAdvisor} encodes sets of accommodations on \texttt{Tripadvisor.com} that are reviewed by the same user~\cite{veldt2021higher}. See Table~\ref{tab:larger} for hypergraph statistics. We also consider hypergraphs from the UCI repository~\cite{asuncion2007uci} that have been use frequently in previous work on hypergraph learning~\cite{panli_submodular,zhang2017re,zhu2022hypergraph,hein2013total}.


\subsection{Comparison against convex relaxations}
\label{sec:approx}
\begin{figure}[t]
	\centering
	\subfigure[Mathoverflow approximation\label{fig:1} ]
	{\includegraphics[width=.495\linewidth]{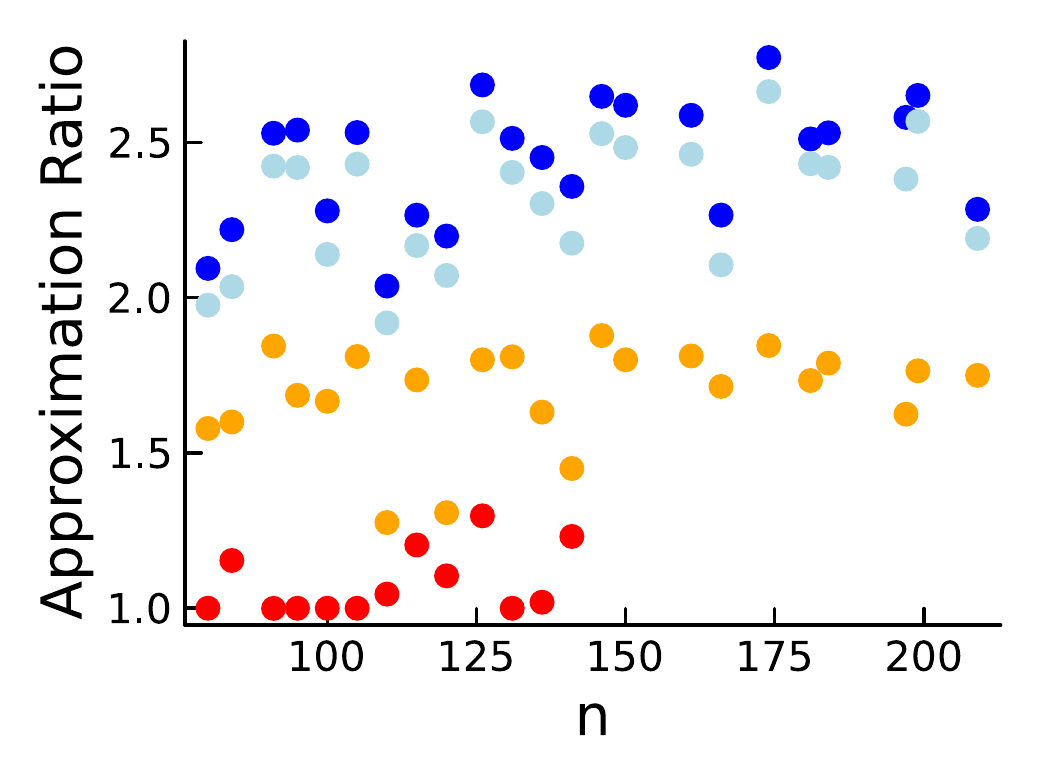}}
	\subfigure[Trivago approximation\label{fig:2} ]
	{\includegraphics[width=.495\linewidth]{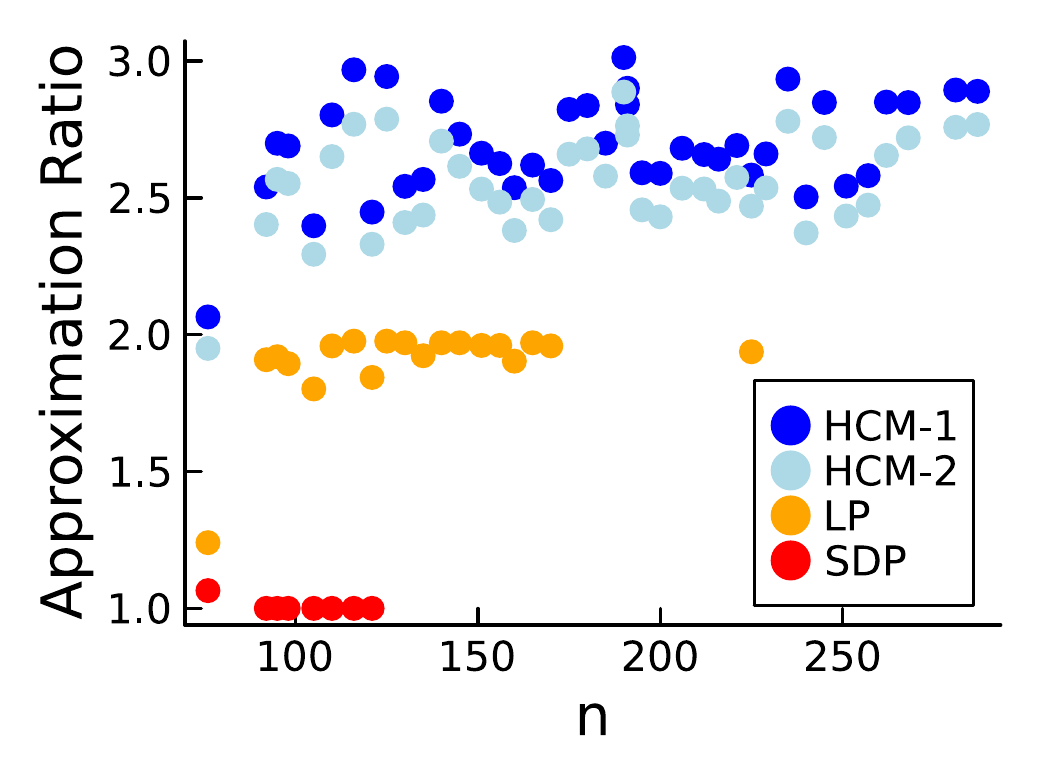}}
	\subfigure[Mathoverflow runtime\label{fig:3} ]
	{\includegraphics[width=.495\linewidth]{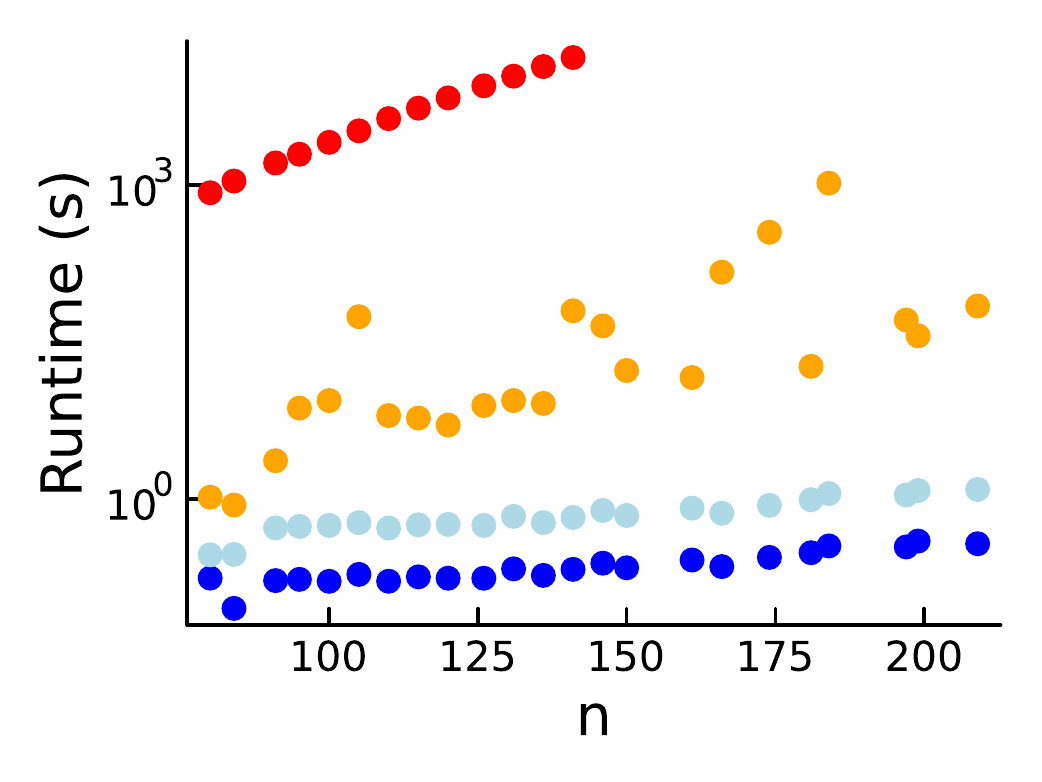}}
	\subfigure[Trivago runtime\label{fig:4} ]
	{\includegraphics[width=.495\linewidth]{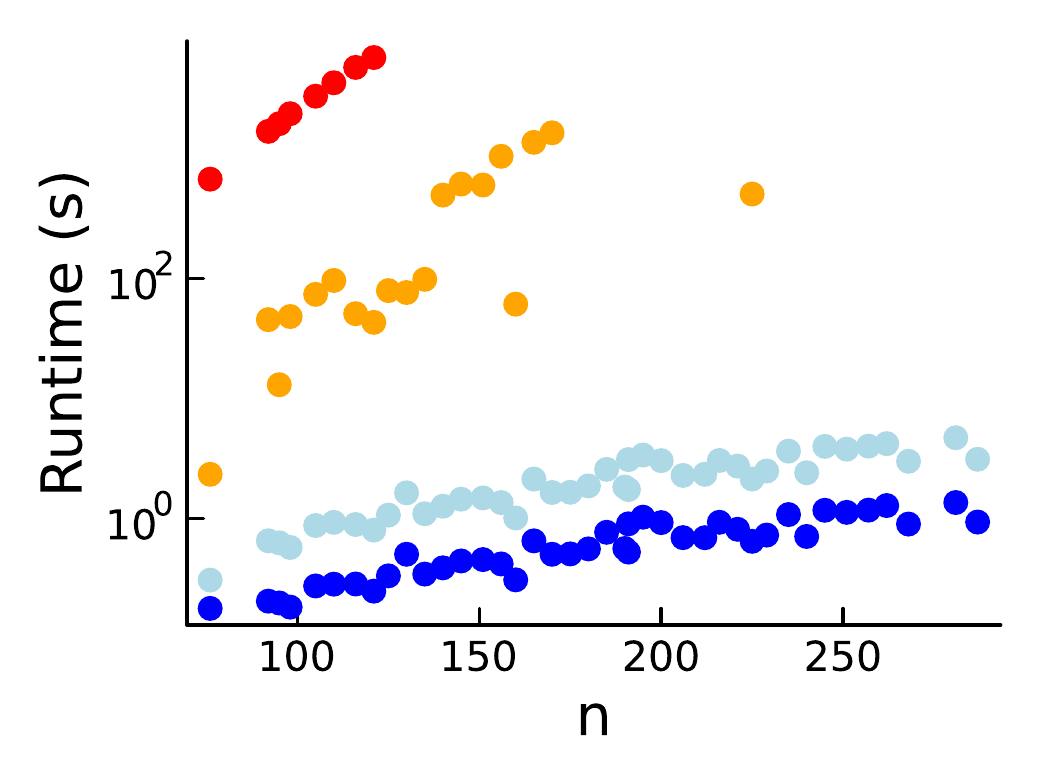}}
	\vspace{-1.0\baselineskip}
	\caption{Results for three hypergraph expansion approximation algorithms on a range of small hypergraphs. Previous approximation algorithms (LP and SDP) cannot scale to larger hypergraphs and often fail even for these small hypergraphs. HCM-1 uses $10 \log_2 n$ rounds of cut-matching, and HCM-2 uses $30 \log_2 n$ rounds.}
	\label{fig:small}
	\vspace{-.95\baselineskip} 
\end{figure}
\begin{table}[t]
	\caption{
		Approximation factor (obtained via Theorem~\ref{thm:bounds}) and runtime for HCM (run for $5 \log_2 n$ iterations) on four hypergraphs, averaged over 5 trials ($\pm$ standard deviation).
	}
	\vspace{-.5\baselineskip} 
	\label{tab:larger}
	\centering
	\scalebox{0.95}{
		\begin{tabular}{lccccc}
			\toprule
			&  $n$ &  $m$ & avg $|e|$ &  Approx. &  Run. (s) \\
			\midrule
		\emph{Amazon9} & 13138 &  31502 &  8.1 & 2.78 {\small $\pm 0.012$} & 254.8 {\small $\pm 10.4$}\\
		\emph{Mathoverfl}&  73851 &  5446 &  24.2 & 3.08 {\small $\pm 0.014$} & 643.8 {\small $\pm 31.1$}\\
		\emph{Tripadvisor}&  8929 &  130568 &  4.1 & 2.71 {\small $\pm 0.025$} & 699.3 {\small $\pm 29.8$}\\
		\emph{Trivago}&  172738 &  233202 &  3.1 & 2.78 {\small $\pm 0.038$} & 2372.6 {\small $\pm 50.8$}\\
			\bottomrule	
		\end{tabular}
	}
	
	\vspace{-.5\baselineskip} 
\end{table} 
\begin{figure}[t]
	\centering
	\subfigure[Approximation\label{fig:app} ]
	{\includegraphics[width=.495\linewidth]{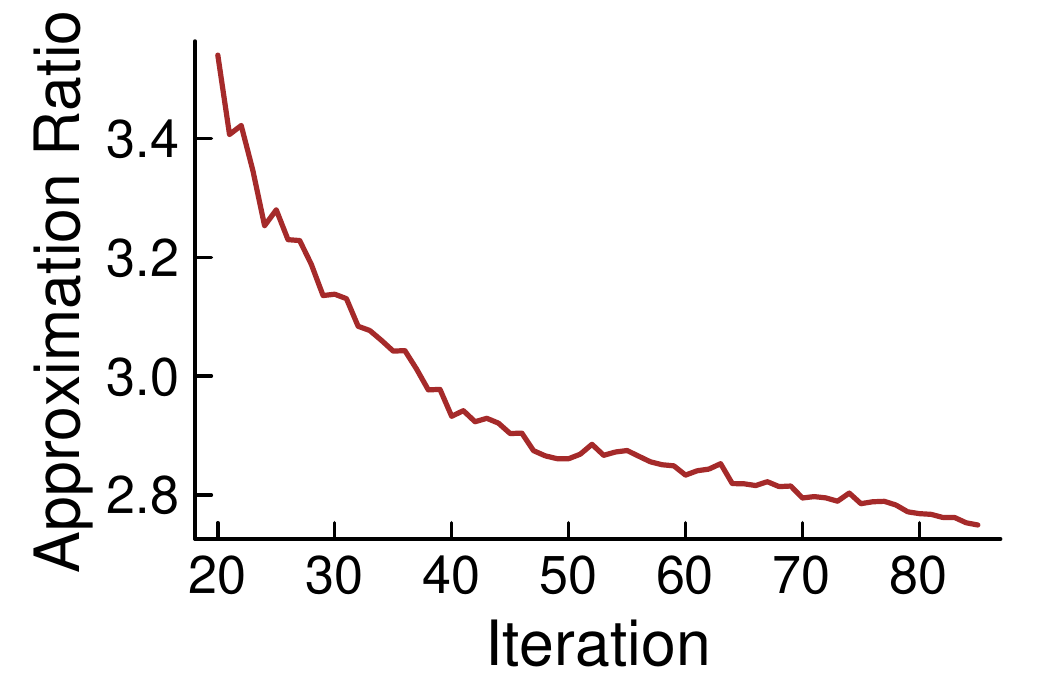}}
	\subfigure[Lower bound\label{fig:lb} ]
	{\includegraphics[width=.495\linewidth]{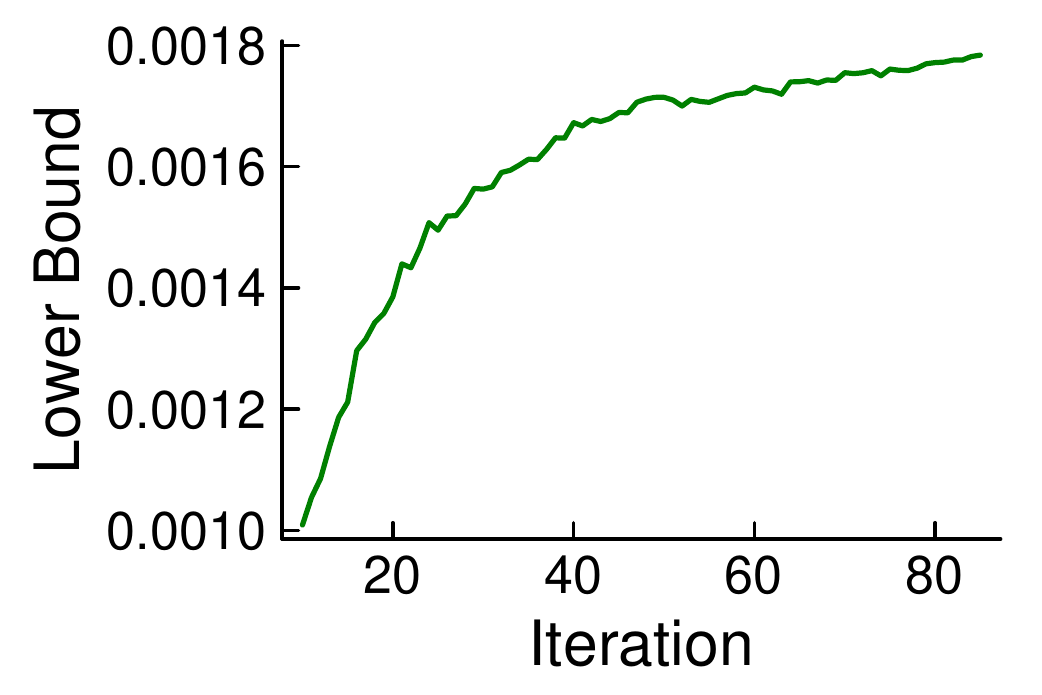}}
	\vspace{-0.5\baselineskip} 
	\caption{The approximation factor and lower bound computed by HCM at each step on the \emph{Trivago} hypergraph.}
	\label{fig:larger}
	\vspace{-0.5\baselineskip} 
\end{figure}

\begin{figure*}[th]
	\centering
	\subfigure[Trivago-Australia Approx Ratios \newline  $n = 1854$, $m = 4328$, avg $|e| = 2.9$ \label{fig:t1} ]
	{\includegraphics[width=.24\linewidth]{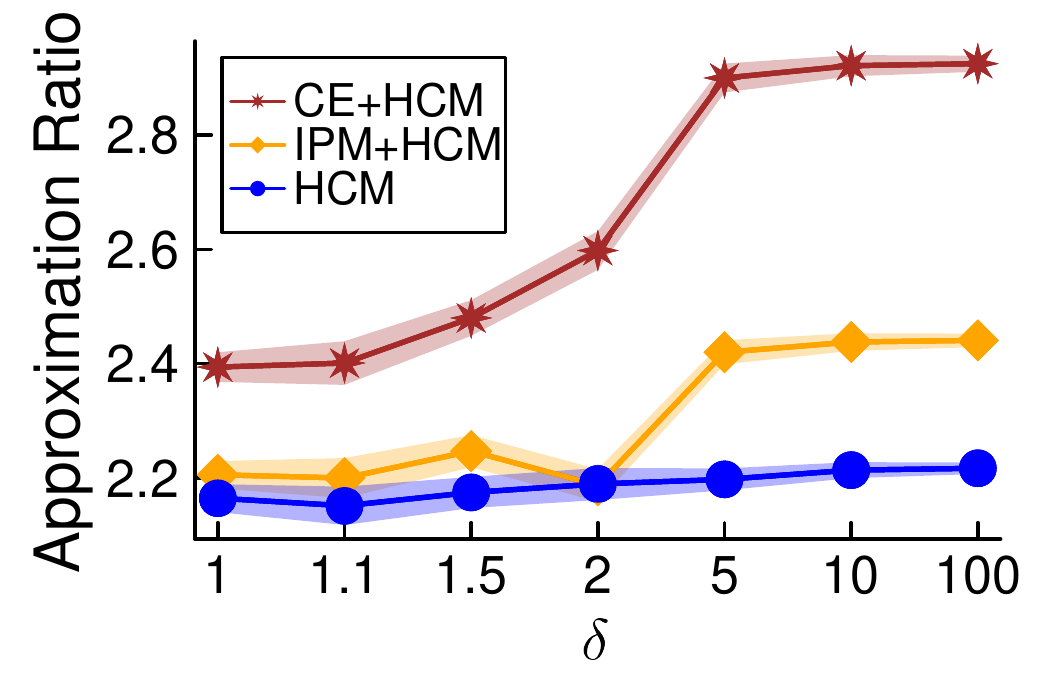}}
	\subfigure[Trivago-UK Approx Ratios \newline $n = 3293$, $m = 9988$, avg $|e| = 3.0$ \label{fig:t2} ]
	{\includegraphics[width=.24\linewidth]{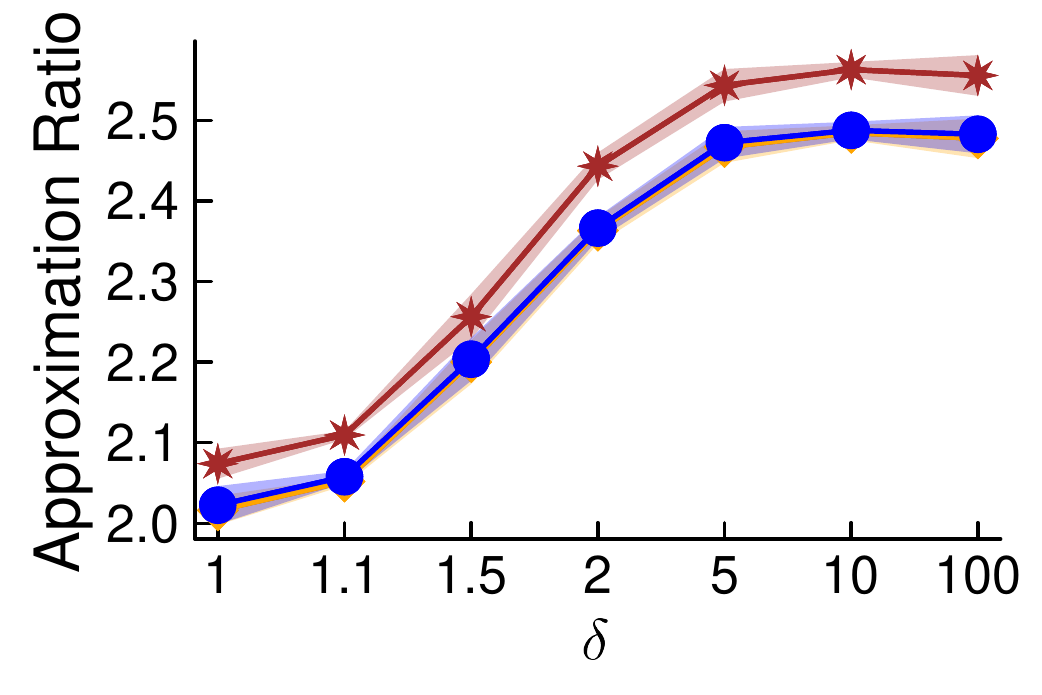}}
	\subfigure[Trivago-Germany Approx Ratios \newline $n = 2869$, $m = 7405$, avg $|e| = 3.0$ \label{fig:t3} ]
	{\includegraphics[width=.24\linewidth]{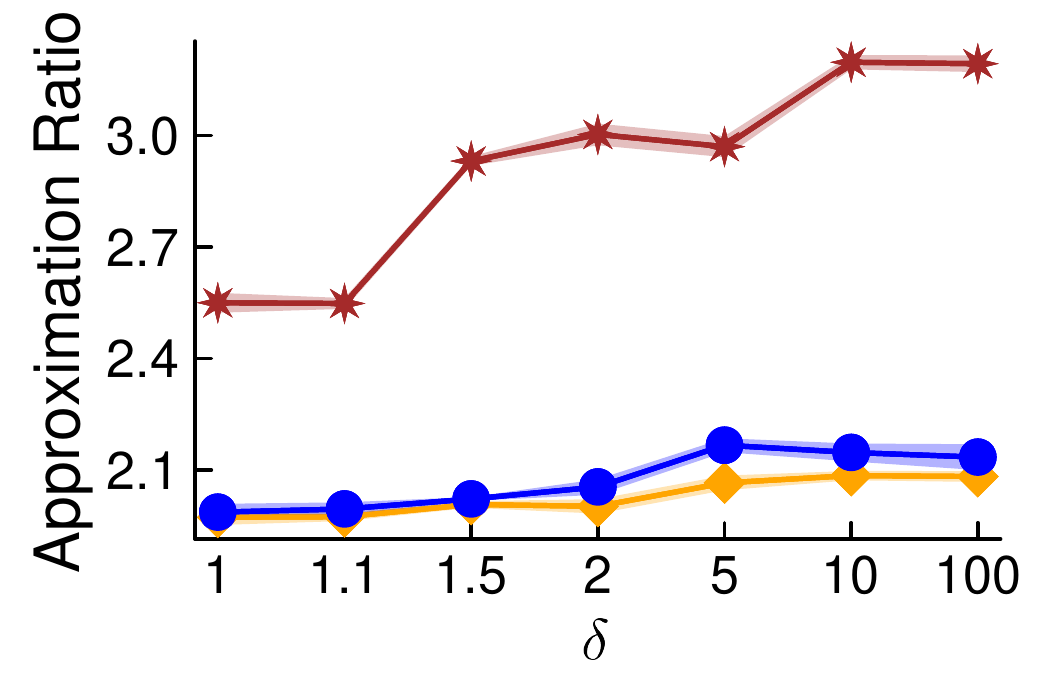}}
	\subfigure[Trivago-Japan Approx Ratios \newline $n = 5598$, $m = 17509$, avg $|e| = 3.0$ \label{fig:t4} ]
	{\includegraphics[width=.24\linewidth]{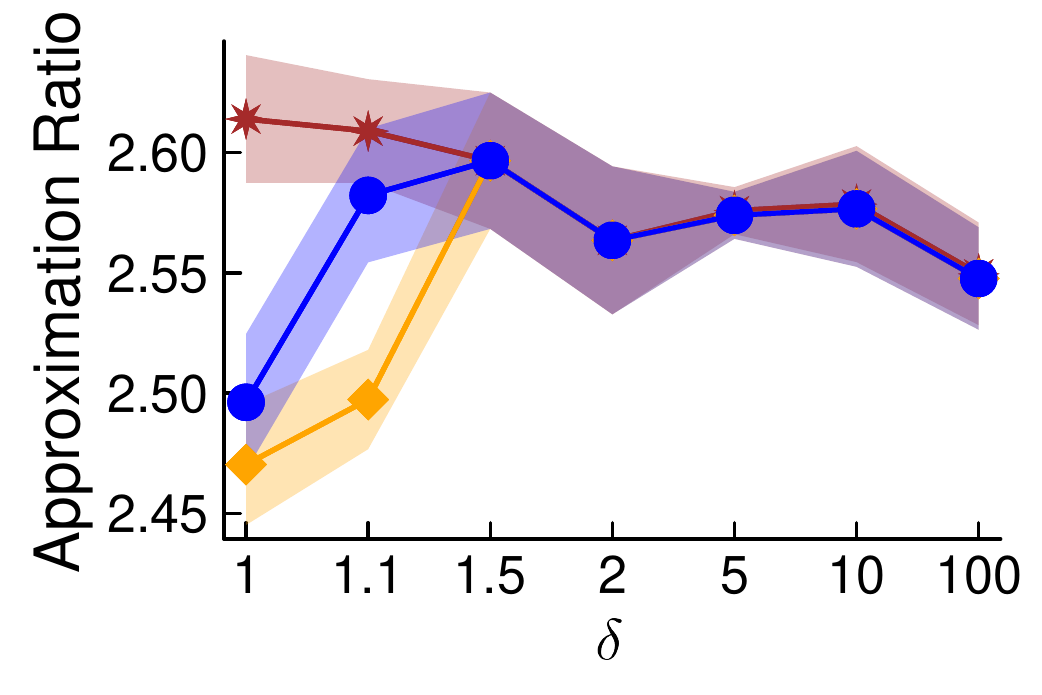}}
	\subfigure[Trivago-Australia Runtimes\label{fig:r1} ]
	{\includegraphics[width=.24\linewidth]{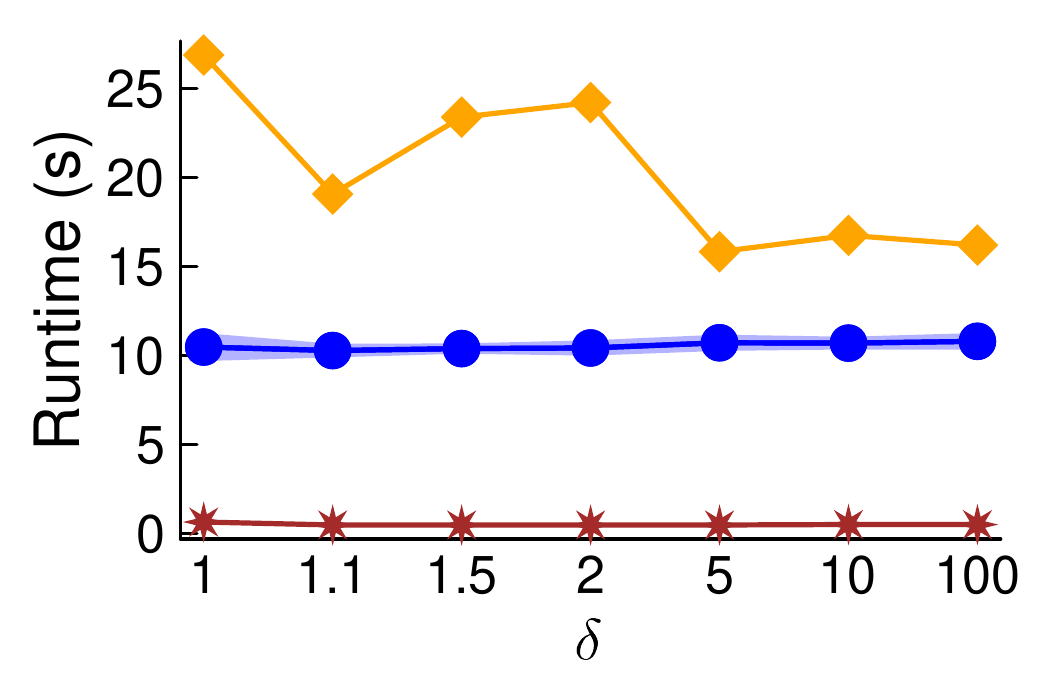}}
	\subfigure[Trivago-UK Runtimes\label{fig:r2} ]
	{\includegraphics[width=.24\linewidth]{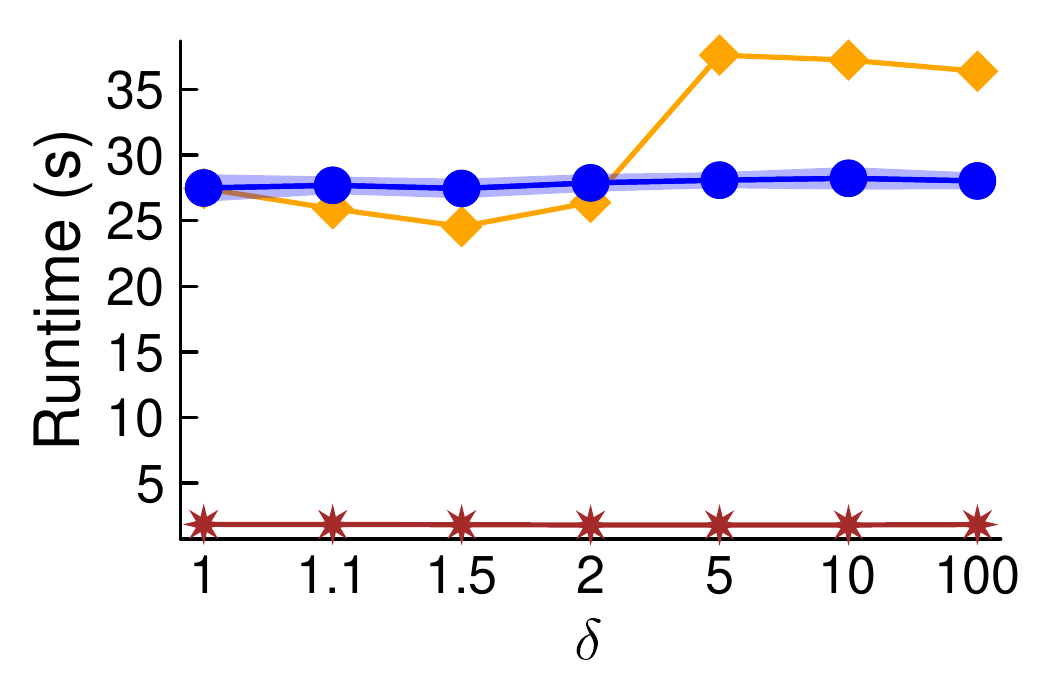}}
	\subfigure[Trivago-Germany Runtimes\label{fig:r3} ]
	{\includegraphics[width=.24\linewidth]{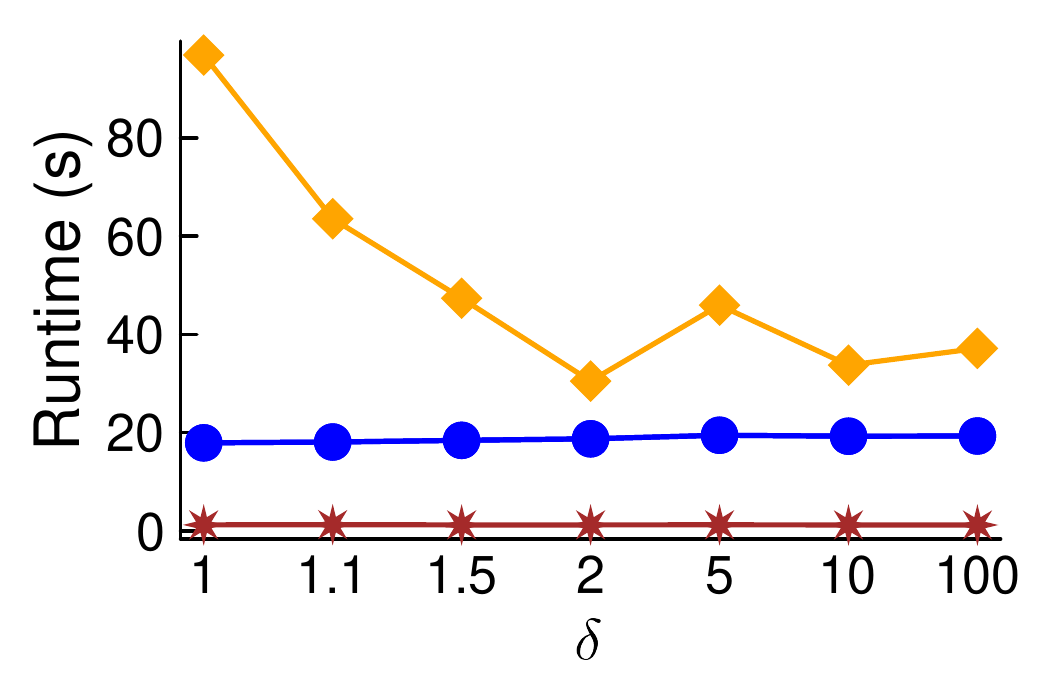}}
	\subfigure[Trivago-Japan Runtimes\label{fig:r4} ]
	{\includegraphics[width=.24\linewidth]{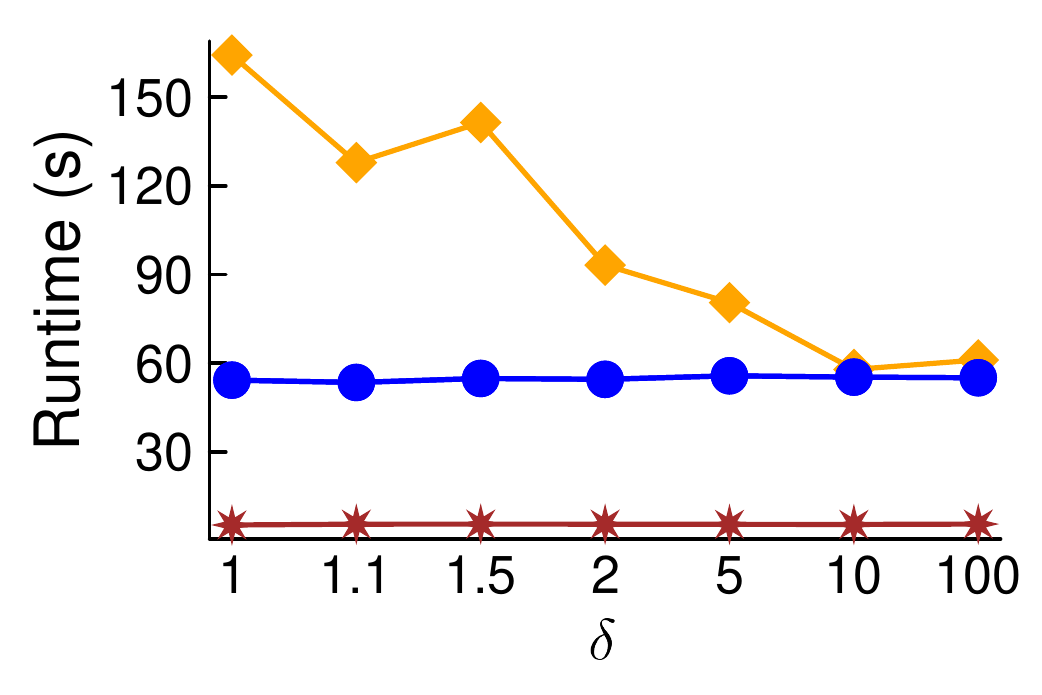}}
	\vspace{-0.5\baselineskip} 
	\caption{Top row: approximations obtained by comparing the best conductance set found by each method against the conductance lower bound computed by HCM. Bottom row: runtimes in seconds. The runtimes for IPM and CE do not include the time it takes to compute the HCM lower bound. We plot mean over 5 runs of HCM; shaded region indicates standard deviation.}
	\label{fig:trivago}
	\vspace{-0.5\baselineskip} 
\end{figure*}

Our method is orders of magnitude more scalable than previous theoretical approximation algorithms for hypergraph ratio cuts and also applies to a much broader class of problems. We compare against the $O(\sqrt{\log n})$ approximation (SDP) based on semidefinite programming~\cite{louis2016approximation}  and the $O(\log n)$-approximation algorithm (LP) based on rounding a linear programming relaxation~\cite{kapralov2020towards}. We focus on the expansion objective ($\pi(v) = 1$ for all $v$) with the standard hypergraph cut function, since these methods do not apply to generalized hypergraph cuts. 

We select a range of small \emph{Mathoverflow}  and \emph{Trivago} hypergraphs of increasing size to use as benchmarks. These correspond to subhypergraphs of the \emph{Mathoverflow}  and \emph{Trivago} hypergraphs in Table~\ref{tab:larger}, each obtained by selecting nodes with a specific metadata label. We chose a set of 23 \emph{Mathoverflow} hypergraphs, each corresponding to a set of posts (nodes) with a certain topic tag (e.g., \emph{graph-theory}). These hypergraphs have between $n = 80$ and $n = 209$ nodes. The number of hyperedges $m$ tends to be between $n/3$ and $n/2$ for these hypergraphs, and the average hyperedge size ranges from 3 to 7. We also consider 41 \emph{Trivago} hypergraphs (each corresponding to a different city tag, e.g., \emph{Austin, USA}), most of which have between 1-3 times as many hyperedges as nodes, and all of which have average hyperedge sizes between $2$ and $4$. 

Figure~\ref{fig:small} reports results for LP, SDP, and our method HCM, all of which compute an explicit lower bound on the optimal expansion which can be used to compute an a posteriori approximation guarantee. We run HCM with two different iteration numbers to illustrate the tradeoff between runtime and approximation guarantee. When SDP and LP converge, they produce very good lower bounds for hypergraph expansion and can be rounded to produce very good a posteriori approximation guarantees. The issue is that these methods do not scale well even to very small instances. We are able to obtain results for all \emph{Mathoverflow} datasets using LP, but the method times out ($>$ 30 minutes) on all but 19 of the 41 \emph{Trivago} hypergraphs. When we remove the time limit, the method sometimes finds good solutions at the expense of much larger runtimes (e.g., over 5 hours for a hypergraph with 229 nodes), but in some cases the underlying LP solver returns without having converged. The SDP method is even less scalable and would not converge for almost any of the small hypergraphs if we set a 30-minute time limit. It took over an hour to run this method for a \emph{Trivago} hypergraph with 121 nodes, and 4.5 hours for a 141-node \emph{Mathoverflow} hypergraph. Given scalability issues, these are the largest hypergraphs of each type that we attempted to address using this method. Overall, even though LP and SDP produce very good lower bounds when they succeed in solving the underlying convex relaxation, they are not practical. Meanwhile, HCM obtains high quality solutions extremely quickly, typically within a matter of a few seconds.

It is worth noting that the difference in approximation ratios in Figures~\ref{fig:1} and~\ref{fig:2} is due almost entirely to the difference in lower bounds computed by each method. These methods return almost identical expansion values, and in many cases find the same output set. This indicates that finding sets with small expansion is not the primary challenge. Rather, the main point of comparison is how close each method's lower bound is to the optimal solution.

\subsection{Expansion on larger hypergraphs}
Table~\ref{tab:larger} lists runtimes and approximations obtained by HCM on four hypergraphs that are orders of magnitude larger than any hypergraph that LP and SDP can handle. Our method is able to quickly obtain solutions within a small factor of the optimal hypergraph expansion solution. Figure~\ref{fig:larger} provides a closer look at how the lower bound and approximation factor maintained by HCM improves at each iteration, specifically for the \emph{Trivago} hypergraph. Plots for the other three datasets look very similar. The lower bound at iteration $t$ (Figure~\ref{fig:lb}) is computed using Theorem~\ref{thm:bounds}, and the approximation factor at iteration $t$ (Figure~\ref{fig:app}) is obtained by dividing the best conductance set found in the first $t$ iterations by this lower bound. The best conductance set returned by HCM is usually found within the first few iterations. Therefore, the main thing accomplished by running the algorithm for many iterations is improving the lower bound certificate. Figure~\ref{fig:larger} shows that the algorithm is still making progress after $5 \log_2 n$ iterations, indicating that a stronger lower bound certificate and approximation guarantee could be obtained by running the algorithm for longer.

\begin{figure*}[t]
	\centering
	\subfigure[Newsgroups Approx Ratios 
	\label{fig:q1} ]
	{\includegraphics[width=.24\linewidth]{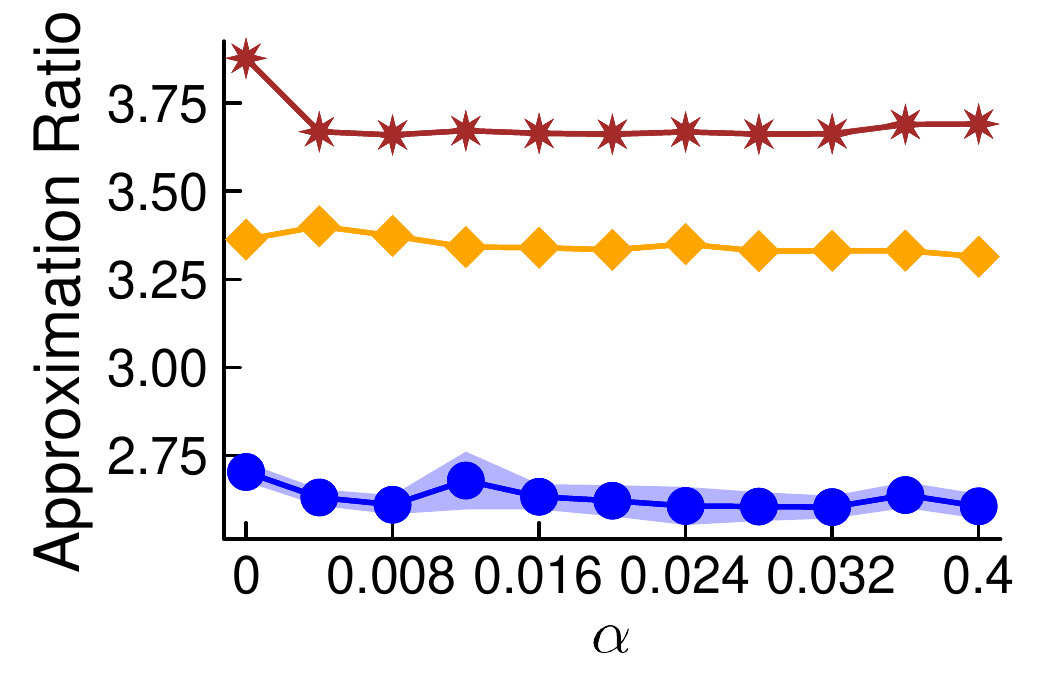}}
	\subfigure[Mushrooms Approx Ratios 
	 \label{fig:q2} ]
	{\includegraphics[width=.24\linewidth]{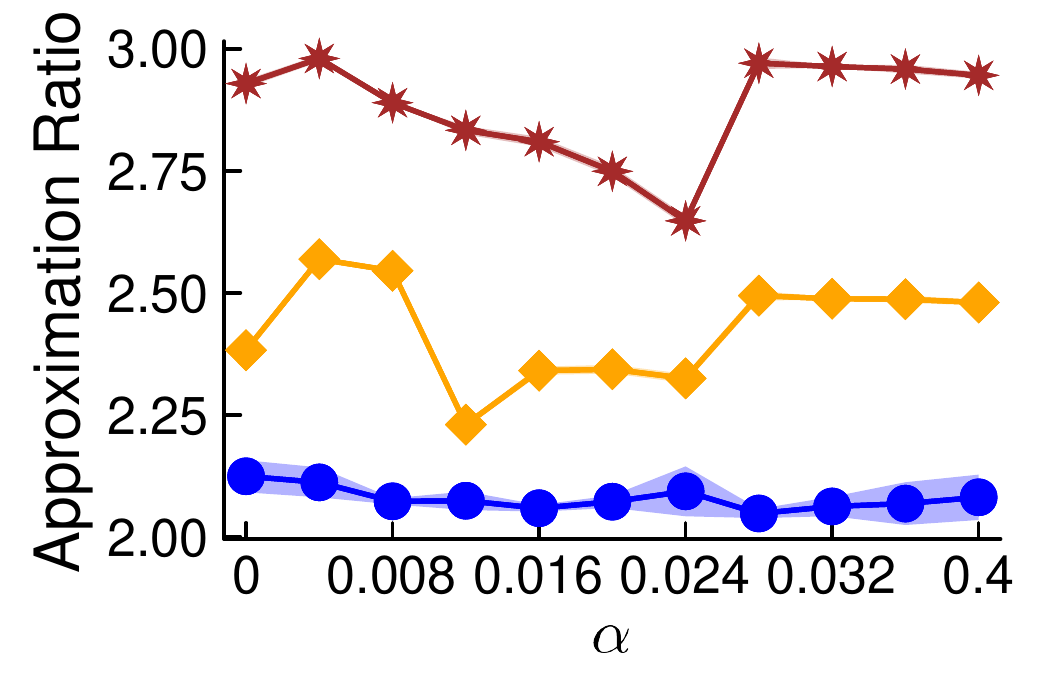}}
	\subfigure[Covertype45 Approx Ratios 
	 \label{fig:q3} ]
	{\includegraphics[width=.24\linewidth]{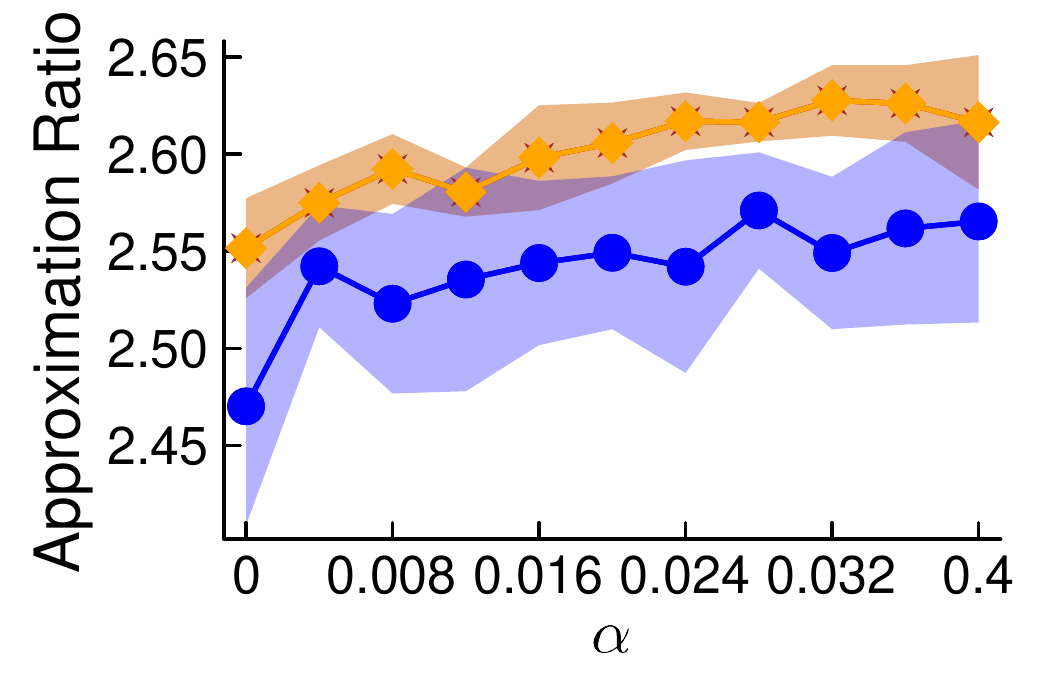}}
	\subfigure[Covertype67 Approx Ratios
\label{fig:q4} ]
	{\includegraphics[width=.24\linewidth]{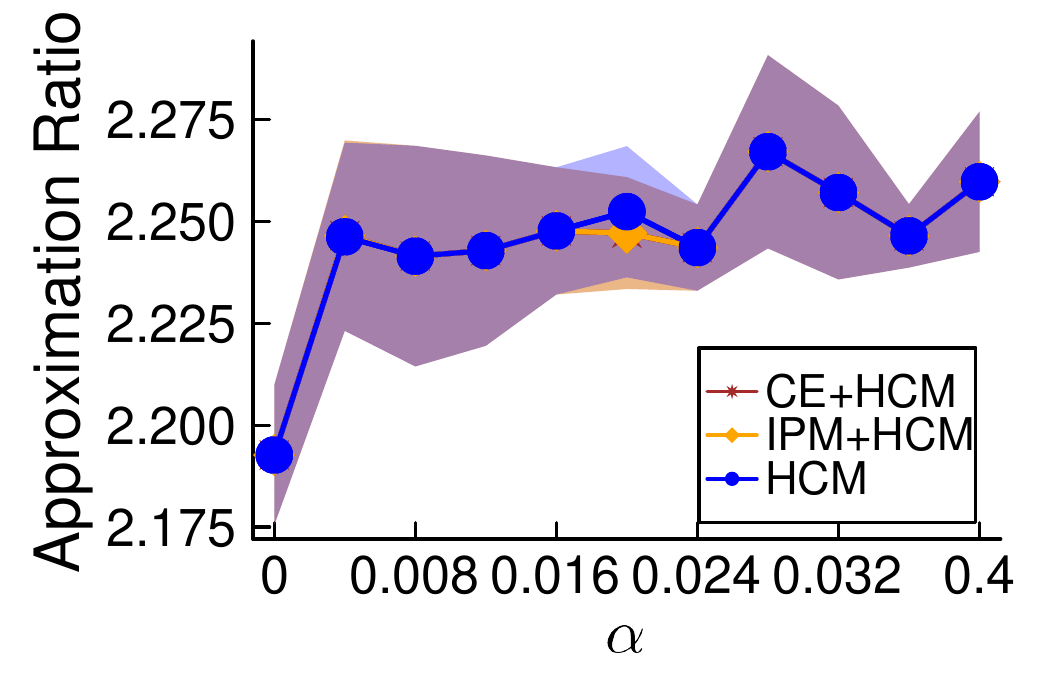}}
	\subfigure[Newsgroups Runtimes\label{fig:p1} ]
	{\includegraphics[width=.24\linewidth]{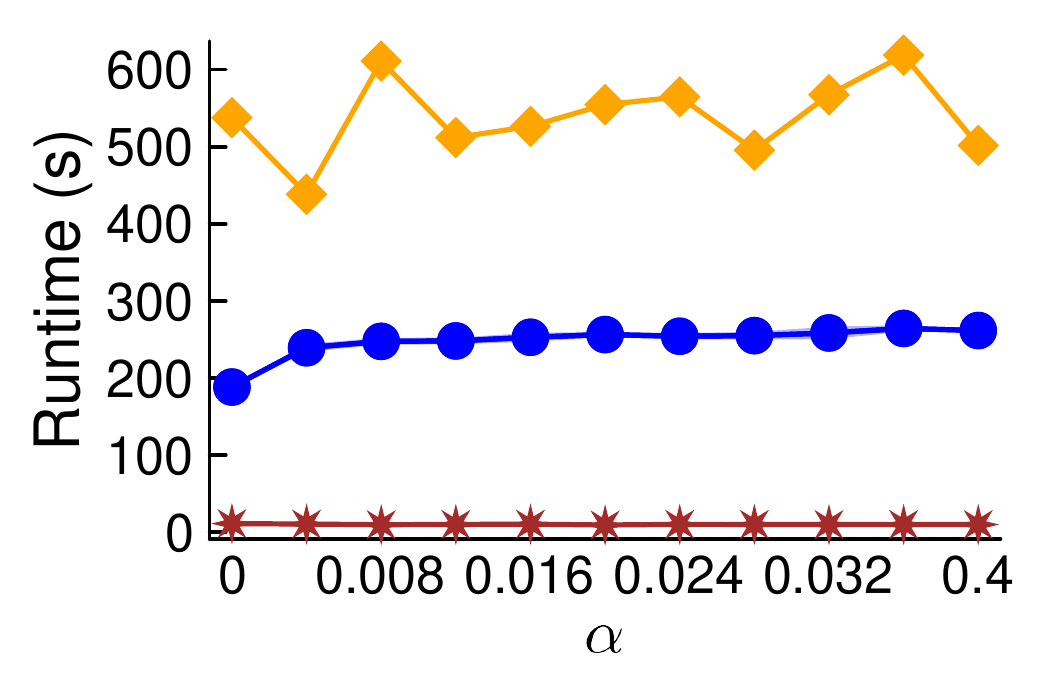}}
	\subfigure[Mushrooms Runtimes\label{fig:p2} ]
	{\includegraphics[width=.24\linewidth]{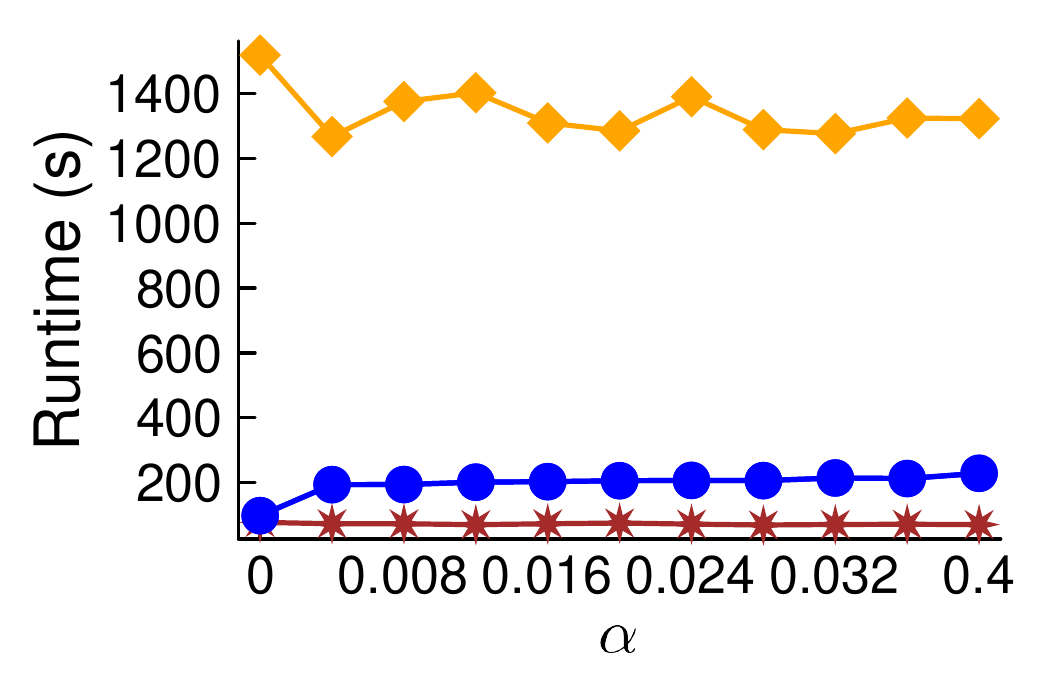}}
	\subfigure[Covertype45 Runtimes\label{fig:p3} ]
	{\includegraphics[width=.24\linewidth]{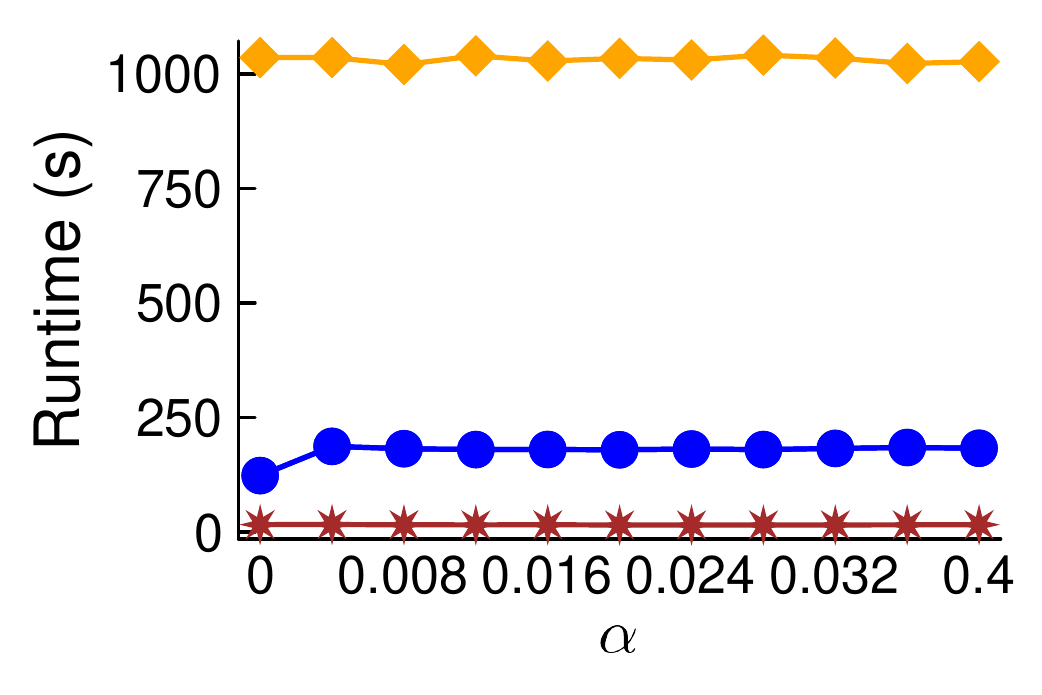}}
	\subfigure[Covertype67 Runtimes\label{fig:p4} ]
	{\includegraphics[width=.24\linewidth]{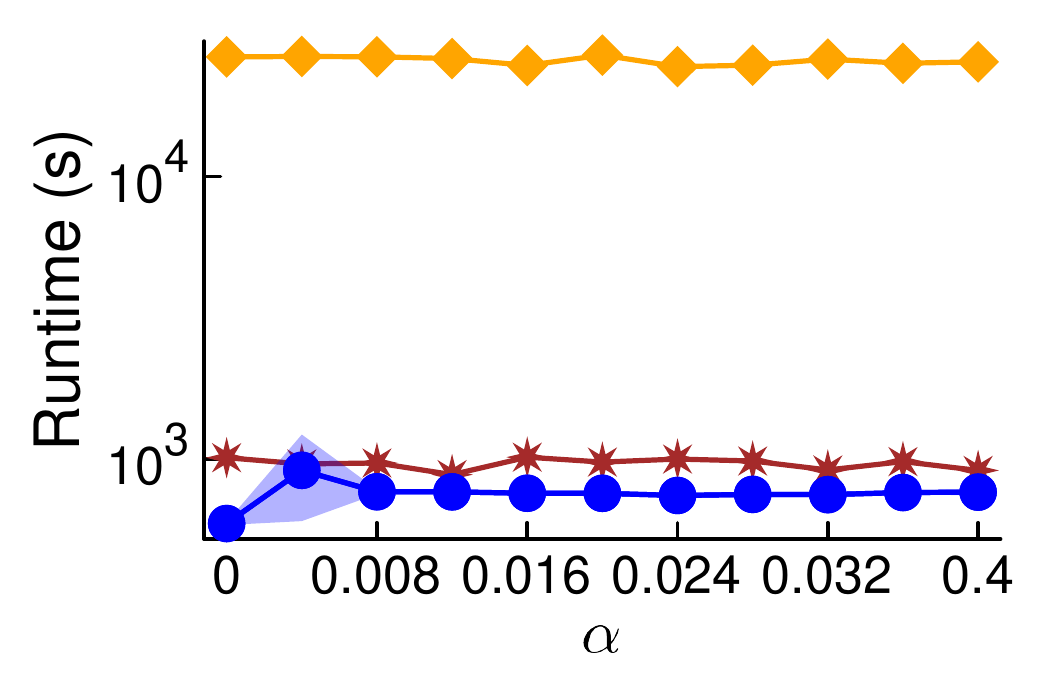}}
	\vspace{-0.5\baselineskip} 
	\caption{Approximation factors and runtimes for running three hypergraph ratio cut algorithms on four benchmark hypergraphs. The runtimes for IPM and CE do not include the time it takes to compute the HCM lower bound. We plot mean over 5 runs of HCM; shaded region indicates standard deviation.}
	\label{fig:benchmark}
	\vspace{-0.5\baselineskip} 
\end{figure*}

\subsection{Trivago hypergraphs and generalized cuts}
\label{sec:trivago}
One distinct advantage of HCM is that it computes explicit lower bounds on hypergraph $\pi$-expansion (via Theorem~\ref{thm:bounds}) that can be used to check a posteriori approximation guarantees in practice, even for other hypergraph ratio cut algorithms that do not come with approximation guarantees of their own. 
To illustrate the power of this feature, we compare HCM against the inverse power method for submodular hypergraphs (IPM) ~\cite{panli_submodular} and the clique expansion method for inhomogeneous hypergraphs (CE)~\cite{panli2017inhomogeneous} in minimizing ratio cut objectives in \emph{Trivago} hypergraphs. We also show how combining algorithms can in some cases lead to the best results.

IPM and CE constitute the current state-of-the-art in minimizing ratio cuts in hypergraphs with generalized cut functions. IPM is a generalization of a previous method that applied only to the standard hypergraph cut~\cite{hein2013total}. CE can be viewed as a generalization of previous clique expansion techniques~\cite{Zhou2006learning,BensonGleichLeskovec2016}, which only applied to more restrictive hypergraph cut functions. Both methods are more practical than LP and SDP, in addition to applying to generalized hypergraph cuts, but they have weaker theoretical guarantees. IPM is a heuristic with no approximation guarantees, while the approximation guarantee for CE scales poorly with the maximum hyperedge size, and can be $O(n)$ in the worst case even for graphs. In our experiments we used the C++ implementation of IPM previously provided by Li and Milenkovic~\cite{panli_submodular}. The performance of IPM depends heavily on which vector it is given as a warm start. We tried several options and found that the best results were obtained by using the eigenvector computed by CE as a warm start. We provide our own new Julia implementation of CE, which is significantly faster than the original proof-of-concept implementation. The appendix includes additional implementation details.

Figure~\ref{fig:trivago} displays runtimes and a posteriori approximation guarantees for HCM, IPM, and CE on four Trivago hypergraphs, corresponding to vacation rentals in Australia, United Kingdom, Germany, and Japan. We specifically consider the 2-core of each hypergraph, as these 2-cores have more interesting and varied cut structure and therefore serve as better case studies for comparing algorithms for generalized hypergraph ratio cuts. For these experiments we are minimizing the conductance objective ($\pi(v) = d_v$), and we use a generalized hypergraph cut function that applies a $\delta$-linear splitting function $\textbf{w}_e(S) = \min \{|S \cap e|, |\bar{S} \cap e|, \delta \}$ at each hyperedge. We run experiments for a range of different choices of $\delta$. We choose this splitting function as previous research has shown that the choice of $\delta$ can significantly affect the size and structure of the output set and influence performance in downstream clustering applications~\cite{veldt2020hyperlocal,veldt2021approximate,liu2021strongly}. The four \emph{Trivago} hypergraphs exhibit a range of different cuts that are found by varying $\delta$, and therefore provide a good case study for how well these algorithms find different types of ratio cuts for generalized splitting functions. 


In terms of finding small conductance sets, HCM trades off in performance with IPM, though they return very similar results. However, HCM is noticeably faster, and it is additionally computing lower bounds that allow it to certify how close its solution is to optimality. 
Our implementation of CE is faster than HCM,
but performs worse in terms of approximation ratio. In Figure~\ref{fig:trivago}, we use the HCM lower bound to obtain a posteriori approximations for IPM and CE. These methods are unable to provide such strong approximation guarantees on their own. Hence, even in cases where IPM finds better conductance sets, the lower bounds computed by HCM provide new information that can be used to prove an approximation guarantee. Since the performance of IPM also relies on using CE as a warm start, we see that often it is a combination of all three algorithms that leads to the best results. 

\subsection{Other benchmark hypergraphs}
\label{sec:benchmark}
We run similar experiments on a set of hypergraphs that have been frequently used as benchmarks for hypergraph learning and clustering~\cite{panli_submodular,Zhou2006learning,hein2013total,zhang2017re}. Statistics for these hypergraphs are shown in Table~\ref{tab:benchmark}.
\begin{table}
	\caption{Statistics for hypergraphs in Figure~\ref{fig:benchmark}.}
	\label{tab:benchmark}
		\begin{tabular}{lcccc}
	\toprule
Hypergraph	&  $n$ &  $m$ & avg $|e|$ & $\sum_e |e|$ \\
	\midrule
\emph{Newsgroups} &  16242 & 100 & 654.5 & 65451 \\ 
\emph{Mushrooms} &  8124 & 112 & 1523.2 & 170604 \\ 
\emph{Covertype45} &  12240 & 128 & 1147.5 & 146880 \\ 
\emph{Covertype67} &  37877 & 136 & 3342.0 & 454514 \\ 
	\bottomrule	
\end{tabular}
	\vspace{-0.5\baselineskip} 
\end{table}
Figure~\ref{fig:benchmark} displays runtimes and approximation factors obtained using the following generalized splitting function
\begin{equation}
	\label{eq:limilen}
	\textbf{w}_e(A) = \frac{1}{2} + \frac12 \min \left\{1, \frac{|S|}{\ceil{\alpha |e|} }, \frac{|e\backslash S|}{\ceil{\alpha |e| }}  \right\}.
\end{equation}
This splitting function was introduced and considered by Li and Milenokovic~\cite{panli_submodular} as a way to better handle categorization errors and outliers. It assigns a smaller penalty to splitting hyperedges in such a way that most of the nodes from a hyperedge are on the same side of the cut. On \emph{Covertype67}, all three methods find essentially the same results, but on the other three hypergraphs HCM shows considerable improvement over IPM and CE in terms of approximation factor. It is also consistently faster than IPM, and on \emph{Covertype67} is even slightly faster than CE.

	\section{Discussion}
	We have presented the first algorithm for minimizing hypergraph ratio cuts that simultaneously (1) has an approximation guarantee that is $O(\log n)$, (2) applies to generalized hypergraph cut functions, and (3) comes with a practical implementation. In practice this algorithm is very successful at finding ratio cuts that are within a small factor (around 2-3) of the lower bound on the optimal solution that the algorithm produces. One open question is to explore how to choose the best generalized hypergraph cut functions to use in different applications of interest. Another open direction is finding improved approximation algorithms for hypergraph ratio cuts that apply to general submodular splitting functions, even those that are not \emph{cardinality-based}. 
	
	\bibliographystyle{plain}
	\bibliography{arxiv-full-updated}
	
	\appendix
	
\section{Proof of Lemma~\ref{lem:match}}
\begin{figure*}[t]
	\centering
	\includegraphics[width=.9\linewidth]{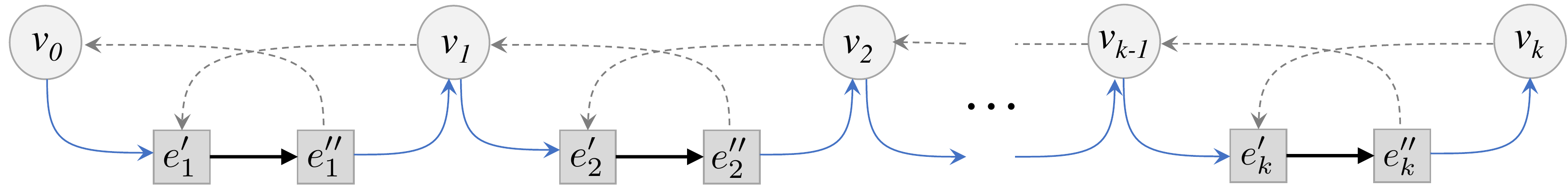}
	\caption{Every directed path in the reduced graph $G(\mathcal{H})$ between two nodes in $V$ alternates between traveling through nodes in $V$ and traveling through pairs of auxiliary nodes $\{e_i', e_i''\}$ from different CB-gadgets. If the original flow function sends flow from $v_0$ to $v_k$, replacing the flow along blue solid edges with flow along the dashed gray edges will reverse the flow direction. It is possible to prove that reversing flow paths in this way will not increase the congestion of the multicommodity flow.}
	\label{fig:path}
\end{figure*}

\noindent \textbf{Proof.}
	Let $f_R = \sum_{(u,v) \in R \times \bar{R}} f_R^{(uv)}$ be the directed multicommodity flow in $G(\mathcal{H})$ obtained from the maximum $s$-$t$ flow $f$ in $G(\mathcal{H},R,\alpha)$. In other words, for each pair $(u,v)$ where $u \in R$ and $v \in \bar{R}$, $f_R^{(uv)}$  is a $u$-$v$ flow function that sends $|f_R^{(uv)}| = \mM_R(u,v)$ flow from $u$ to $v$. Each $u \in R$ will be a \emph{deficit} node, and each $v \in \bar{R}$ will be an \emph{excess} node for this multicommodity flow function.
	
	We can assume without loss of generality that $f$ is cycle-free, which implies that $f_R$ is also cycle-free.
	Let $S \subseteq V$ be an arbitrary set; our goal is to edit $f_R$ to turn it into a flow function $f_S =  \sum_{(u,v) \in S \times \bar{S}} f^{(uv)}_S$ that routes ${\mM_R}(u,v)$ flow from $u$ to $v$ whenever $u \in S$, $v \in \bar{S}$, and $\{u,v\}$ is an edge in $M_R$. In particular, this means we need to find a way to reverse the flow direction for any pair $(u,v) \in (S\cap \bar{R}) \times (\bar{S} \cap R)$, since for this type of pair $f_R$ sends flow from $v$ to $u$ rather than from $u$ to $v$.
	
	We will first consider what it means to reverse flow on a single directed flow path. Consider an arbitrary edge $\{u,v\}$ in $M_R$ such that $u \in S\cap\bar{R}$ and $v \in \bar{S} \cap R$. The original multicommodity flow function $f_R$ includes a flow function $f_R^{(vu)}$ that sends $\mM_R(u,v)$ units of flow from $v$ to $u$. By Lemma~\ref{lem:flowdeomp}, this can be decomposed as $f_R^{(vu)} = \sum_{j = 1}^r f_j^{(vu)}$, where each $f_j$ is a directed flow path from $v$ to $u$.  By the construction of $G(\mathcal{H})$, this flow path will travel through nodes and edges in a sequence of CB-gadgets corresponding to a sequence of hyperedges $\{e_1, e_2, \hdots, e_k\} \subseteq \mathcal{E}$ in the hypergraph $\mathcal{H} = (V, \mathcal{E})$. For the $i$th hyperedge $e_i$ in this sequence and the corresponding $i$th CB-gadget, let $e_i'$ be the first auxiliary node in the gadget for this hyperedge, and $e_i''$ be the second. Recall that there is an edge $(e_i', e_i'')$, and for each node $x \in e_i$ there is a directed edge $(x, e_i')$ and a directed edge $(e_i'',x)$ (see Figures~\ref{fig:gadget} and~\ref{fig:path}). A simple directed flow path $f_j^{(vu)}$ from $v = v_0$ to $u = v_{k+1}$ is therefore is a sequence of the form:
	\begin{align}
		\label{eq:onedir}
		\begin{array}{l}
			v = v_0 \rightarrow (e_1' \rightarrow e_1'')  \rightarrow v_1 \rightarrow (e_2' \rightarrow e_2'')  \rightarrow v_2 \rightarrow \cdots  \\
			\hspace{.3cm} \cdots \rightarrow v_{k-1} \rightarrow (e_k' \rightarrow e_k'')  \rightarrow v_{k} = u,
		\end{array}
	\end{align}
	where $v_t \in V$ for each $t \in \{0,1, \hdots, k\}$, and where the same amount of flow is sent along each edge. By the construction of every CB-gadget, there is a path in the opposite direction with the same exact set of edge weights:
	\begin{align}
		\label{eq:twodir}
		\begin{array}{l}
			u = v_k \rightarrow (e_k' \rightarrow e_k'')  \rightarrow v_{k-1} \rightarrow \cdots \\
			\hspace{.3cm} \cdots \rightarrow v_2 \rightarrow (e_2' \rightarrow e_2'')\rightarrow   v_1 \rightarrow  (e_1' \rightarrow e_1'') \rightarrow  v_0 = v. 
		\end{array}
	\end{align}
	Recall that $f_R$ is cycle-free. This means that no flow is sent along edges of the form $(v_i, e_i')$ for $i \in \{1,2, \hdots, k\}$, because $f_R$ contains a flow path $f_j^{(vu)}$ with positive flow on the edges $(e_i', e_i'')$ and $(e_i'', v_i)$. Therefore, we can replace the flow path $f_j^{(vu)}$ in~\eqref{eq:onedir} with the flow path $\hat{f}_j^{(vu)}$ defined on the edges in~\eqref{eq:twodir}, with the same flow value as $f_j^{(vu)}$ but traveling in the opposite direction. Figure~\ref{fig:path} illustrates this flow reversal process for a single flow path.
	
	We construct a new multicommodity flow function $f_S$ by simultaneously applying this flow reversal process to every flow path in $f_R$ that goes in the wrong direction. Formally, if $p_{vu}$ denotes the set of indices for simple directed flow paths from $v$ to $u$ in $f_R$ (obtained from a flow decomposition as in Lemma~\ref{lem:flowdeomp}), then $f_R = \sum_{(v,u) \in R\times \bar{R} } \sum_{j \in p_{vu}} f_j^{(vu)}$, and we define:
	\begin{equation*}
		f_S = f_R + \sum_{(v,u) \in (\bar{S}\cap R) \times (S \cap \bar{R}) } \sum_{j \in p_{vu}} \hat{f}_j^{(vu)} - f_j^{(vu)}.
	\end{equation*}
	In other words, for every $(u,v) \in (S\cap \bar{R}) \times (\bar{S} \cap R)$ and every directed flow path $f_j^{(vu)}$ from $v$ to $u$ in $f_R$, replace $f_j^{(vu)}$ with $\hat{f}_j^{(vu)}$. 
	
	It remains to prove that $f_S$ has congestion at most $1/\alpha$. We prove this by considering different types of edges $(v, e')$, $(e'', v)$, and $(e', e'')$, where $v \in V$, and where $\{e', e''\}$ are the two auxiliary nodes from a CB-gadget for some hyperedge $e \in \mathcal{E}$. Observe first of all that this flow reversal procedure never changes the flow on edge $(e', e'')$, so the congestion remains the same on edges that go between auxiliary nodes. The other two edges $(v,e')$ or $(e'', v)$ have the same weight, and because $f_R$ is cycle free, at most one of these edges has a positive amount of flow in $f_R$. If $(v,e')$ has a positive amount of flow in $f_R$, then $(e'',v)$ has no flow in $f_R$. If the flow reversal process changes anything, it will take some of the flow through $(v,e')$ and transfer it to $(e'', v)$. In this case, the congestion on edge $(v,e')$ cannot be worse because we are only removing flow. Meanwhile, the edge $(e'',v)$ started with no flow and then received some of the flow from $(v,e')$. After the flow transfer process, the congestion on $(e'',v)$ will not exceed $1/\alpha$ since it has the same weight as edge $(v,e')$ and the congestion on $(v,e')$ was at most $1/\alpha$. We can provide an analogous argument in the case where $(e'',v)$ has a positive flow in $f_R$ but $(v,e')$ does not. Thus, this flow reversal process does not make the congestion worse.\\

\section{Experiments and Implementations}

This section of the appendix provides additional details regarding algorithm implementation and experimental results.

\subsection{Algorithms based on convex relaxations}
In Section~\ref{sec:approx} we compare our method against the previous SDP-based method~\cite{louis2016approximation} and a previous LP-based method~\cite{kapralov2020towards}. For these experiments we focus on the all-or-nothing hypergraph cut penalty
and the standard expansion objective ($\pi(v) = 1$ for every $v \in V$). All hypergraphs we consider are unweighted.

\paragraph{Implementation details for the LP method.} Kapralov et al.~\cite{kapralov2021towards} primarily focused on spectral sparsifiers for hypergraphs; details for the $O(\log n)$ approximation algorithm for hypergraph sparsest cut are included only in the full arXiv version of the paper~\cite{kapralov2020towards}. The algorithm is presented for the conductance objective ($\pi(v) = d_v$ for each $v \in V$) but can be easily adapted to expansion ($\pi(v) = 1$ for each $v \in V$) in order to compare more directly with the SDP approach~\cite{louis2016approximation}. This adaptation requires solving the following LP relaxation over a hypergraph $\mathcal{H} = (V,\mathcal{E})$:
\begin{equation}\label{lp}
	\begin{array}{rll}
		\min & \sum_{e \in \mathcal{E}} z_{e} & \\
		\text{subject to } &  \sum_{u,v \in V} x_{uv} = n & \\
		& z_e \geq x_{uv} & \forall e \in \mathcal{E}; u,v \in e \\
		& x_{uv} \leq x_{uv} + x_{vw} & \forall u,v,w \in V \\
		& x_{uv} = x_{vu} \geq 0 &\forall u,v \in V
	\end{array}		
\end{equation}
One can check that the solution to this LP is a lower bound for the sparsest cut objective:
\begin{equation*}
	\min_{S \subseteq V} \frac{\cut_\mathcal{H}(S)}{|S|} + \frac{\cut_\mathcal{H}(S)}{|\bar{S}|}.
\end{equation*}
Dividing the optimal LP solution value by 2 provides a lower bound for hypergraph expansion:
\begin{equation*}
	\min_{S \subseteq V; |S| \leq |V|/2} \frac{\cut_\mathcal{H}(S)}{|S|}.
\end{equation*}
Recall that this applies only when $\cut_\mathcal{H}$ is the all-or-nothing hypergraph cut penalty. The variables $\{x_{uv}\}$ obtained from solving the LP define a pseudo-metric on the nodes of the hypergraph. These can be rounded into a cut with $O(\log n)$ distortion using an embedding result of Bourgain~\cite{bourgain1985lipschitz}. We implement the algorithm in Julia using Gurobi optimization software to solve the LP relaxation. Chapter 11 of Trevisan's lecture notes~\cite{trevisan2017lecture} provides a clear and convenient overview of Bourgain's embedding theorem, which we follow in our implementation of the rounding method.

\paragraph{Implementation details for the SDP method.}
We solve the SDP relaxation for hypergraph expansion~\cite{louis2016approximation} with Mosek optimization software, using CVX~\cite{cvx} in Matlab as a front end. For better efficiency, CVX automatically converts the primal program into a dual SDP before solving the dual. The most time consuming aspect of this procedure by far is the time it takes to convert to the dual. We also tried calling the Mosek SDP solver using Julia as a front end \emph{without} first converting the problem to the dual, but this performs significantly worse. We also tried using the SDPT3 solver in Matlab, but this also performed worse and was unable to find a solution for many problems that Mosek was able to solve.

The rounding procedure of Louis and Makarychev~\cite{louis2016approximation} for the SDP relaxation is very complicated and not designed for practical implementation. The details of this algorithm span multiple technical lemmas and theorems within the paper, and also depend on black-box results from previous papers, which in turn depend on unspecified constants that are needed for the rounding procedure. We therefore used a simple rounding scheme that still produced good results in practice. Solving the SDP relaxation produces a vector $\textbf{x}_u$ for each node $u \in V$, which allows us to compute distance variables $d(u,v) = \| \textbf{x}_u - \textbf{x}_v\|_2^2$ and similarity scores $s(u,v) = \textbf{x}_u^T \textbf{x}_v$ for each pair of nodes $(u,v) \in V \times V$. For each node $v \in V$, we then checked the expansion of the set of nodes \emph{closest} or \emph{most similar} to $v$ with respect to the functions $d$ and $s$. More precisely, for each $v \in V$, we checked the expansion of sets of the form
\begin{align*}
	B(v,r) = \{u \colon d(u,v) < r\} \text{ and } M(v,r) = \{u: s(u,v) > r\},
\end{align*}
for each value of $r \in \mathbb{R}^+$ that produced a different set. Although this is much more simplistic than the theoretical rounding scheme, it often produced expansion values that matched the SDP lower bound, or were at least below a factor of 1.3 from this lower bound. 

\subsection{Generalized hypergraph cut experiments}
In Sections~\ref{sec:trivago} and~\ref{sec:benchmark} we compare our method against IPM and CE, two other methods that handle generalized hypergraph cuts in ratio objectives. 
Both methods apply to the following definition of hypergraph degree  for a node $v \in V$:
\begin{equation}
	\label{gendegree}
	d_v = \sum_{e \colon v \in e} \textbf{w}_e(\{v\} ).
\end{equation}
For the $\delta$-linear splitting function (with $\delta \geq 1$), this definition of hyperedge degree is equivalent to counting the number of hyperedges that a node participates in, though this is not the case for the splitting functions in~\eqref{eq:limilen}. Our method is able to accommodate any node weights, so we also use~\eqref{gendegree} to facilitate comparisons.

The benchmark hypergraphs we consider in Section~\ref{sec:benchmark} all come from the UCI machine learning repository (\url{https://archive.ics.uci.edu/ml/datasets.php}). 
There are a few minor variations in the number of nodes, number of hyperedges, and the sum of hyperedge sizes for these hypergraphs reported in different papers~\cite{panli_submodular,hein2013total,zhang2017re,zhu2022hypergraph}, likely due to minor differences in how the original dataset was converted into a hypergraph. We were unable to find processed versions of these hypergraphs, so we reconstructed them from the original UCI datasets. Our hypergraphs are included in our GitHub repository.


\paragraph{Implementation details for IPM and CE}
We use the previous C++ implementation of IPM~\cite{panli_submodular}. We compute warm start vectors separately using our implementation of CE, and this is not included in the runtimes reported for IPM. 

The clique expansion method for inhomogeneous hypergraphs (CE) works by replacing each hyperedge by a weighted clique in a way that minimizes the difference between cut penalties in the clique and generalized cut penalties of the original hyperedge splitting function~\cite{panli2017inhomogeneous}. If $\textbf{w}_e$ denotes the hyperedge splitting function for $e \in \mathcal{E}$ and $\hat{\textbf{w}}_e$ denotes the cut penalties in the projected clique, the goal is to choose edge weights in the clique so that
\begin{align}
	\label{eq:ce}
	\textbf{w}_e(A) \leq \hat{\textbf{w}}_e(A) \leq C \textbf{w}_e(A) \text{ for all $A \subseteq e$},
\end{align}
for the smallest possible value of $C$. Spectral clustering is then applied to the reduced graph. The first implementation of this method for generalized hypergraph cuts was designed for specific cut penalties that are not cardinality-based (\url{https://github.com/lipan00123/InHclustering}, accompanying~\cite{panli2017inhomogeneous}). A later implementation, in Matlab, was designed for specific cardinality-based submodular splitting functions (\url{https://github.com/lipan00123/IPM-for-submodular-hypergraphs/clique_exp.m}~\cite{panli_submodular}). 

We provide an implementation of CE in Julia that works for arbitrary cardinality-based submodular splitting functions. For cardinality-based functions, all edges in the clique expansion of a hyperedge are given the same weight $p$, so that $\hat{\textbf{w}}_e(A) = p |A| |e \backslash A|$. By checking all possible values of $|A|$ from $1$ to $|e|$, one can easily find the optimal value of $p$ such that~\eqref{eq:ce} holds for the minimum value of $C$. Once the hypergraph is projected into a graph, there are several slight variants of spectral clustering that can be applied in the reduced graph. When computing a normalized Laplacian matrix, one can normalize based on degrees in the reduced graph, or based on the hypergraph degrees. Our code allows for either. Figures~\ref{fig:trivago} and~\ref{fig:benchmark} display results for normalizing based on graph degrees, which more closely follows the original theory of Li and Milenkovic~\cite{panli2017inhomogeneous}. Results are very similar when normalizing based on hypergraph degrees. After computing the second smallest eigenvector $\textbf{v}_2$ of the normalized Laplacian, spectral clustering checks sweep cuts of the form $\{u \in V \colon \textbf{v}_2(u) > r\}$, for all values $r$ leading to a different cut set. The theoretical guarantees for CE are based on taking the sweep cut minimizing conductance in the reduced graph. Following the previous implementation of Li and Milenkovic~\cite{panli_submodular}, we select the sweep cut with smallest conductance in the original hypergraph. Checking conductance in the graph is faster, as it is easier to update the graph cut function when iterating through all sweep cuts, but checking conductance in the original hypergraph leads to better results in practice. For the Trivago experiments, our Julia implementation of CE is 40-50 times faster than the previous Matlab implementation. We did not check the runtime for the Matlab implementation on the hypergraphs in Section~\ref{sec:benchmark}.


\end{document}